\documentclass[12pt]{article}
\usepackage[english]{babel}
\usepackage[margin=1.0in]{geometry}
\usepackage[utf8]{inputenc}
\usepackage[pdftex]{hyperref}
\usepackage{amsmath,amsthm,amsfonts,amssymb,graphicx,float,caption,booktabs, 
subcaption,setspace,tikz,authblk,verbatim,natbib,indentfirst,sectsty, lipsum}
\usepackage{titlesec}
\titlespacing*{\section}{0pt}{.2\baselineskip}{.2\baselineskip}

\usepackage{rotating}
\newcommand{\dr}{\textsuperscript{\#}}

\usetikzlibrary{arrows,shapes}
\usepackage{xcolor}
\hypersetup{
  colorlinks   = true, 
  urlcolor     = blue, 
  linkcolor    = blue, 
  citecolor   =  black 
}

\allsectionsfont{\centering}
\definecolor{dkred}{rgb} {.84,0.37,0.00}

\tikzset{
    mark_label/.style={label={[black]below:#1}},
    treenode/.style = {align=center, inner sep=0pt, text centered,
            font=\sffamily,text width=1.3em},
    arn_w/.style = {treenode, circle, black, draw=black,
            text width=1.6em},
    arn_w1/.style = {treenode, circle, black, draw=black, fill = dkred,
            text width=1.6em},
    arn_w2/.style = {treenode, circle, black, draw=black, fill = dkred, text width=0em}
}

\newenvironment{theorem}[2][Theorem]{\begin{trivlist}
\item[\hskip \labelsep {\bfseries #1}\hskip \labelsep {\bfseries #2.}]}{\end{trivlist}}
\newenvironment{lemma}[2][Lemma]{\begin{trivlist}
\item[\hskip \labelsep {\bfseries #1}\hskip \labelsep {\bfseries #2.}]}{\end{trivlist}}

\newcommand{\bl}{\begin{flushleft}}
\newcommand{\el}{\end{flushleft}}
\newcommand{\bc}{\begin{center}}
\newcommand{\ec}{\end{center}}
\newcommand{\bE}{{\mathbb{E}}}

\newcommand{\bI}{{\mathbb{I}}}

\newcommand{\mH}{{\mathcal{H}}}
\newcommand{\mA}{{\mathcal{A}}}
\newcommand{\mB}{{\mathcal{B}}}

\title{\bf\large A Bottom-up Approach to Testing Hypotheses That Have a Branching Tree Dependence Structure, \\with False Discovery Rate Control\\ }

\author[1]{Yunxiao Li}
\author[1]{Yi-Juan Hu\thanks{corresponding author}}
\author[2]{Glen A. Satten}
\affil[1]{{\normalsize Department of Biostatistics and Bioinformatics, Emory University}}
\affil[2]{{\normalsize Centers for Disease Control and Prevention}}
\date{}

\makeatletter
\def\blfootnote{\xdef\@thefnmark{}\@footnotetext}
\makeatother

\begin{document}

\begin{titlepage}
\maketitle
\thispagestyle{empty}

\vspace{-1.1cm}

\doublespacing

\vspace{-0.2cm}
\noindent Modern statistical analyses often involve testing large numbers of hypotheses.  In many situations, these hypotheses may have an underlying tree structure that not only helps determine the order that tests should be conducted but also imposes a dependency between tests that must be accounted for.  Our motivating example comes from testing the association between a trait of interest and groups of microbes that have been organized into operational taxonomic units (OTUs) or amplicon sequence variants (ASVs).  Given $p$-values from association tests for each individual OTU or ASV, we would like to know if we can declare that a certain species, genus, or higher taxonomic grouping can be considered to be associated with the trait. For this problem, a bottom-up testing algorithm that starts at the lowest level of the tree (OTUs or ASVs) and proceeds upward through successively higher taxonomic groupings (species, genus, family etc.) is required.  We develop such a bottom-up testing algorithm that controls the error rate of decisions made at higher levels in the tree, conditional on findings at lower levels in the tree. We further show this algorithm controls the false discovery rate based on the global null hypothesis that no taxa are associated with the trait. By simulation, we also show that our approach is better at finding \textit{driver taxa}, the highest level taxa below which there are dense association signals. We illustrate our approach using data from a study of the microbiome among patients with ulcerative colitis and healthy controls.\\

\end{titlepage}

\doublespacing

\newpage

\setcounter{page}{1}

\section{Introduction}
\label{sec:introduction}

The false discovery rate (FDR) has largely replaced the family-wise error rate to control the error made when testing many hypotheses. \cite{benjamini_controlling_1995} proposed a simple way to to control the FDR when testing independent hypotheses, which extends easily to hypotheses having positive regression dependence \citep{benjamini_control_2001}.  However, an adjustment to the Benjamini and Hochberg procedure \citep{benjamini_control_2001} to allow arbitrary dependence between tests is very conservative in most settings.  For this reason, testing procedures that control FDR for specific patterns of dependence have been investigated.

\cite{yekutieli_hierarchical_2008} considered hypotheses tests organized in a branching tree using an approach that starts by testing the hypothesis at the ``top" of the tree; if this hypothesis is rejected, hypotheses at the next lowest level are tested.  Testing continues from top to bottom until no further hypotheses can be rejected, at which point no further tests are conducted.  This approach is appropriate for some problems, such as the motivating example in \cite{yekutieli_hierarchical_2008} in which a genome-wide test for (genetic) linkage was conducted, followed by tests for linkage separately on each chromosome, then tests for linkage on the p or q arms of each chromosome, etc.  Moving down the tree corresponds to increasing localization of the linkage signal, making the top-down strategy a natural choice.  The null hypothesis being tested at each node in the tree is the ``global null hypothesis" that none of the tests below this node are significant.

Some problems are not well-suited to the top-down approach.  For example, in a 16S rRNA microbiome study, bacterial sequences are typically grouped into operational taxonomic units (OTUs) or assigned to amplicon sequence variants (ASVs).  Then, the association between each OTU or ASV and a trait of interest is calculated (for simplicity, we restrict our discussion to OTUs with the understanding that the argument can apply to either ASVs or OTUs).  Typically, many OTUs will belong to the same species.  After testing association between each OTU and a trait, we may wish to determine if any species of microbes are associated with the trait.  Depending on our findings, we may wish to test larger groups corresponding to successively higher taxonomic ranks (e.g., genus, family, order, class, phylum, and kingdom).  The natural ordering of hypotheses in this example starts at the bottom of the tree and proceeds upward.  Further, it may be desirable to continue to test hypotheses at higher levels of the tree even if no findings have been made at a lower level, since an accumulation of weak signals from lower levels may coalesce into a detectable signal at a higher level.  The scientific questions of interest thus motivate development of a bottom-up approach to testing tree-structured hypotheses.

When considering if we should declare a certain node (say, a genus) to be associated with the trait we are studying, we adopt the following approach: if a large proportion of species from that genus influence the trait, we should conclude the genus influences the trait.  Conversely if only a few of the species from a genus are non-null, then a better description of the microbes that influence occurrence of the trait is a list of associated species.  Finding taxa that can be said to influence a trait in this sense is the first goal of our approach.  The second goal is to locate the highest taxa in the tree for which we can conclude many taxa below, but not any ancestors above, influence risk; we refer to such taxa as \textit{driver} taxa.  We also consider a related criterion, the \textit{conjunction} null hypothesis, that would require that \textit{all} species from the genus be associated with the trait before declaring the genus is associated.

The rest of this paper is organized as follows.  In Section \ref{sec:method}, we propose a modified null hypothesis for bottom-up testing that adjusts for selection decisions at lower levels of tree. We further develop an error criterion we call the false assignment rate (FAR) that corresponds to this modified null hypothesis, and propose an algorithm for assessing the significance of association between taxonomic units (or, more generally, nodes in the tree) and a trait under study that controls the FAR.  In Section \ref{sec:simu}, we compare our proposed methods with other existing methods using simulated data, and show that the FAR can approximate the FDR under the conjunction null hypothesis.  In Section \ref{sec:data1}, we apply our new methods to data on the human gut microbiome from a study of inflammatory bowel disease (IBD), and detect driver taxa that are associated with ulcerative colitis (UC).  Section \ref{sec:dis} contains a discussion of our results and some possible future directions implied by our work.

\section{Methods}
\label{sec:method}

\noindent{\large\bf 2.1 Preliminaries}

The hypotheses we test form the nodes of a branching tree; here, we review the terminology we use.  The \textit{root} node is the ``top" of the tree (in Figure \ref{tree}(a), $N_{4,1}$ is the root node).  For any two nodes that are directly connected, the node closest to (furthest from) the root is the \textit{parent }(\textit{child}) node.  The set of child nodes of a parent are its \textit{offspring}. A node is an \textit{inner node} if it has at least one child node; otherwise it is a \textit{leaf} node.  The \textit{ancestors} of a node are all the nodes traversed in a path from that node to the root.  The \textit{descendents} of a node are all nodes having that node as an ancestor.   A \textit{subtree} is a tree rooted at an inner node of the full tree, comprised of the subtree root and all its descendants.    For example, in Figure 1(a), the tree rooted at $N_{3,1}$ that includes inner nodes $N_{3,1}$, $N_{2,1}$, $N_{2,2}$, and leaf nodes $N_{1,1}$, $N_{1,2}$, $N_{1,3}$, $N_{1,4}$ is a subtree of the full tree.  The \textit{depth} of a node is the number of edges between that node and the root node. Because each node corresponds to a hypothesis, we will sometimes refer to testing a node as shorthand for testing the hypothesis at that node.

\begin{figure}
\centering
(a)
\includegraphics[width=1.0\textwidth]{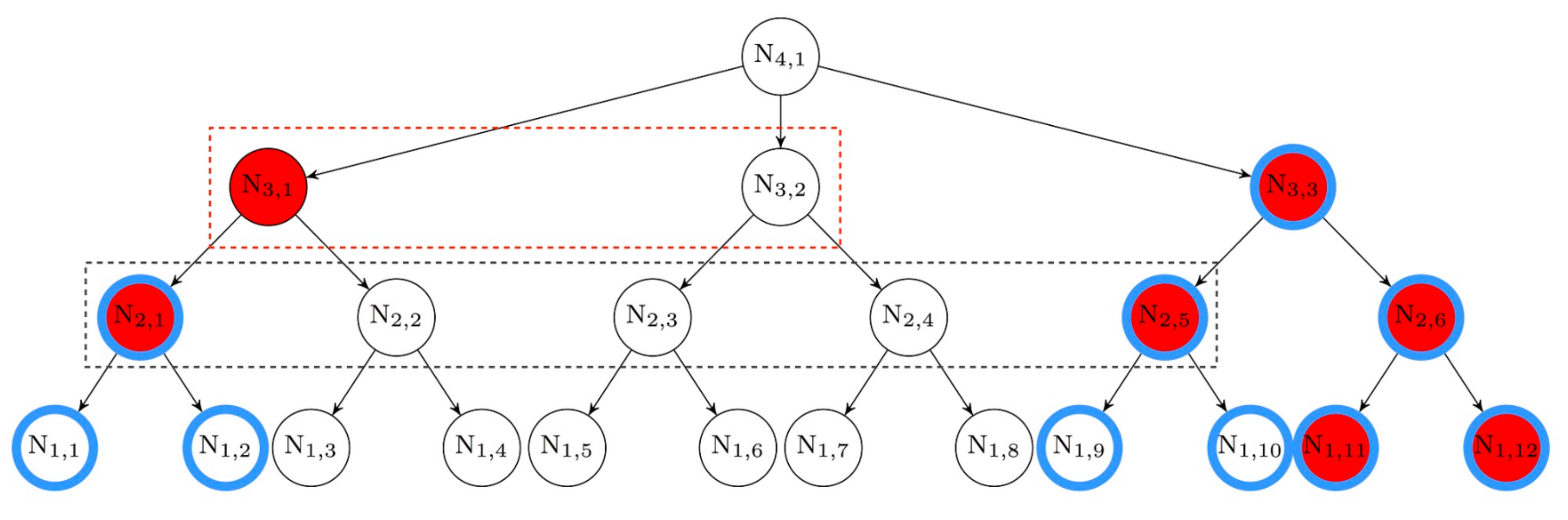}
(b)
\includegraphics[width=1.0\textwidth]{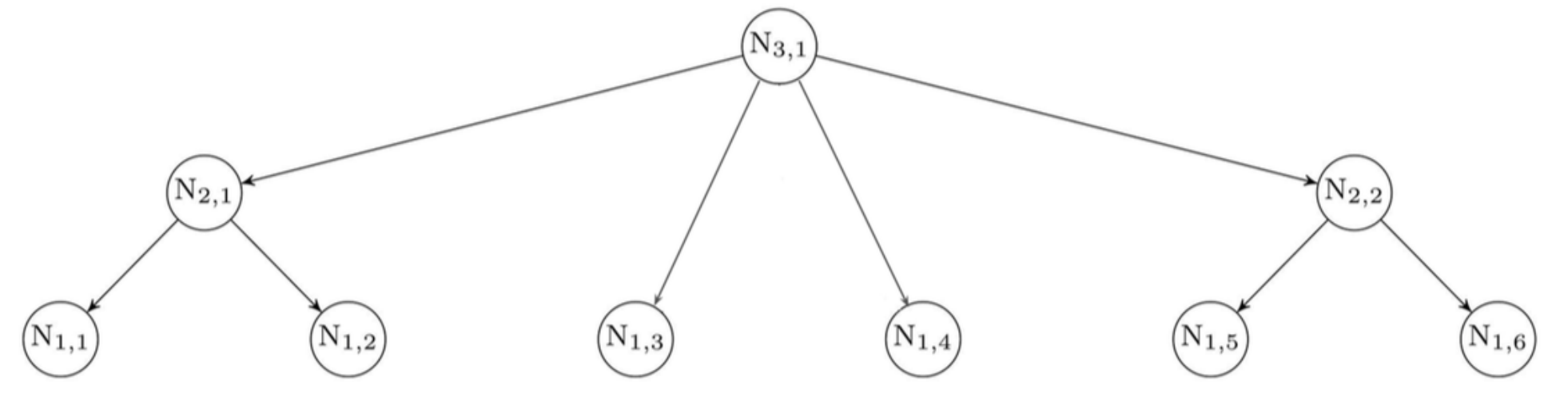}
\caption{\label{tree}{(a) A hypothetical example of a set of hypotheses having a tree structured relationship. Nodes are labeled by \textit{level} (first subscript) and then numbered within level (second subscript).  Nodes highlighted with blue circles are truly associated. A node colored red indicates it is detected (declared to be associated with the trait of interest by a testing method). With the bottom-up methods, all nodes at the bottom level are tested at level 1, nodes inside the black box are tested at level 2, and nodes inside the red box are tested at level 3. (b) A hypothetical example illustrating a set of hypotheses having a dependence structure corresponding to an incomplete tree.  In this example, it makes scientific sense to assign nodes $N_{1,3}$ and $N_{1,4}$ to level 1 even though they have a different depth than the other leaf nodes. For example, these two nodes could correspond to OTUs that are missing a species assignment but share a genus with the other leaf nodes. }}
\end{figure}

We use the term \textit{level} to describe sets of nodes that will be tested together.  In the simple case such as in Figure \ref{tree}(a), nodes that have the same depth are assigned to the same level; we call such a tree \textit{complete}.  For \textit{incomplete} trees such as that shown in Figure \ref{tree}(b), level is assigned by the investigator and does not necessarily correspond to depth. For example, in a phylogenetic (taxonomic) tree, level typically corresponds to the taxonomic rank (species, genus, etc.); a phylogenetic tree is then incomplete when the leaf nodes (OTUs) have missing assignment below a certain level.  The tree shown in Figure \ref{tree}(b) could be an example of this, where level 1 corresponds to OTUs and where OTUs $N_{1,1}, N_{1,2}, N_{1,5}$, and $N_{1,6}$ have genus \textit{and} species assignments but OTUs $N_{1,3}$ and $N_{1,4}$ are \textit{only} assigned at the genus level.  We will sometimes refer to a node at which we have rejected the null hypothesis as a \textit{detected} node, and a detected node is a \textit{driver} node if none of its ancestors are detected.  We further assume that \textit{p}-values for the association tests at all leaf nodes are available.  The algorithm we describe here gives \textit{p}-values for all internal nodes as needed.

There are two ways to imagine calculating $p$-values for inner nodes in a bottom-up testing algorithm for tree-structured hypotheses.  In the first approach, $p$-values for inner nodes are determined entirely from the $p$-values of their offspring (and hence, are determined by the $p$-values at leaf nodes).  In the second approach, $p$-values for inner nodes are calculated by applying a test statistic to pooled data \citep{tang2016general}.  The second approach may be problematic as it may be hard to determine the null distribution of the pooled data, given decisions about $p$-values at lower levels (e.g., removing data from nodes having $p$-values less than a threshold from pooling for testing the modified null defined in Section 2.3). Also, it is hard to know how the conjunction null hypothesis could be tested using pooled data. In addition, pooling data may result in effects being cancelled out if some offspring nodes are protective while others increase risk. For these reasons, we seek an algorithm that operates entirely on the $p$-values of the leaf nodes.

\bigskip

\noindent{\large\bf 2.2 Bottom-Up Testing}

The goals of our inference are to find nodes in the tree (e.g., taxa for the microbiome example) which are associated with a trait of interest.  We wish to avoid declaring a node to be associated just because a few offspring nodes are strongly associated; thus we restrict claims of association to nodes in which a large number of offspring are associated.  For this goal the global null hypothesis, which specifies a node is associated if \textit{even one} offspring node is associated, is not appropriate.  Directly testing the conjunction null hypothesis that \textit{not all} offspring of a node are associated \citep{price_cognitive_1997} is known to be conservative in many situations, as the $p$-value is determined by selecting the \textit{largest} $p$-value from the $p$-values at each offspring, and then comparing the selected $p$-value to the uniform distribution on $[0,1]$ \citep{friston2005conjunction}. However, it does easily lead to a bottom-up procedure; after propagating the largest $p$-value from offspring nodes to their parent nodes, nodes are then detected using the standard BH \citep{benjamini_controlling_1995} procedure.  We report results from this procedure even though we do not recommend it, due to its low power.

To develop a bottom-up procedure that avoids the low power of the rigorous test of the conjunction null hypothesis, but still finds taxa for which most offspring are associated, we use the insight that, if an offspring node has already been found to be associated, including it in our test for the parent adds no new information.  This insight is reinforced by the observation that, in an omnibus test of association for all offspring, a strong association from already-detected children may lead to the parent node being detected even if most other child nodes are truly null.  Thus, as a surrogate for the conjunction null hypothesis, we propose to test a modified null that, among the offspring \textit{that have not been previously detected at a lower level}, none are associated, against the alternative that some previously-undetected offspring are associated.  We then combine the $p$-values of these previously-undetected nodes to form a test statistic for the parent node, using a small modification of Stouffer's Z test.  Nodes for which all offspring are already detected are not tested, but automatically detected and require special handling as described in Section 2.4.  

To illustrate these issues, consider the hypothetical example in Figure \ref{tree}(a).  Nodes highlighted with blue circles are truly associated; $N_{2,1}$ and $N_{3,3}$ are driver taxa. Although $N_{3,1}$ is associated under the global null hypothesis because $N_{1,1}$ and $N_{1,2}$ are its descendants, we would prefer to conclude that $N_{2,1}$ rather than $N_{3,1}$ explains the association signal among the descendants of $N_{3,1}$ because $N_{3,1}$ has descendants that are not truly associated.  This is achieved by using the modified null hypothesis, because $N_{3,1}$ is \textit{not} associated under the modified null hypothesis, because $N_{2,2}$ is not associated.

The error rate of any testing algorithm depends on the null hypothesis used to ascertain the true association status of each node.  Thus, we distinguish between the FDR, for which we use the global null hypothesis to determine true association status; the false assignment rate (FAR), for which we use the modified null to determine true association status; and the FDRc, for which we use the conjunction null to determine true association status.  For all three error rates, a false discovery means that a node was detected (i.e., found to be associated) that is in fact not associated, under the appropriate null. We use the term FAR rather than modified FDR because the ``assignment" of nodes as being associated or not is influenced by decisions made at lower levels, and reserve the term ``discovery" for situations where decisions are based on a test statistic that is calculated for each node without regard to decisions at other nodes. Thus, our procedure must test all nodes at the lowest level, then the next lowest level and so on.   Suppose the tree has $L$ levels, and let $n_l$, $l = 1,2,\ldots,L,$ be the number of nodes on level $l$ and $N_{l,j}$, $j=1,2,\ldots,n_l$, denote the $j^\text{th}$ node on level $l$. We then define
$$
R_{l,j}=\bI\left(\text{node}~N_{l,j}~\textit{is detected}\right),
$$
where $\bI\left(.\right)$ is the indicator function. We note that if node $N_{l,j}$ is not tested, then $R_{l,j}=0$ by default.  Let $\mathcal{U}_{l,j}$ denote the set of undetected offspring of the  $j^\text{th}$ node on level $l=2, \ldots, L$; note that $\mathcal{U}_{1,j}=\emptyset$. We next define $V^x_{l,j}$ to indicate a false assignment was made under null hypothesis $x$, where $x=g$ for the global null hypothesis, $x=m$ for the modified null hypothesis, and $x=c$ for the conjunction null hypothesis. Thus, 
$$
V_{l,j}^{m}=\bI\left(R_{l,j}=1 \text{ but the } \textit{modified} \text{ null hypothesis given }\mathcal{U}_{l,j} \text{ at node }N_{l,j} \text{ is true}\right). 
$$
Similarly, for the global and conjunction null hypotheses, we define $V_{l,j}^g$ and $V_{l,j}^c$ as
$$
V_{l,j}^{x}=\bI\left(R_{l,j}=1 \text{ but null hypothesis } x \text{ at node } N_{l,j} \text{ is true}\right)   , \text{     } x=c,g.
$$

The error rate under each null hypothesis is given by
\begin{equation*}
\mathbb{E}  \left[  \frac{\sum_{l=1}^{L} \sum_{j=1}^{n_l} V_{l,j}^{x}}{\left(\sum_{l=1}^{L} \sum_{j=1}^{n_l} R_{l,j}\right)\bigvee 1}  \right].
\end{equation*}
If the global null hypothesis at $N_{l,j}$ is true, then the modified null hypothesis at $N_{l,j}$ must also be true, which in turn implies the conjunction null hypothesis at $N_{l,j}$ is true.  Thus, $V_{l,j}^{g} \le V_{l,j}^{m} \le V_{l,j}^c$ holds for all nodes. Thus, the three error rates FDR, FAR and FDRc defined above, are related by
$$
\mathrm{FDR} \le \mathrm{FAR} \le \mathrm{FDRc}.
$$

This implies that controlling FAR is a more stringent criterion than controlling FDR, and so a testing procedure that controls the FAR will automatically control the FDR. However, controlling FAR does not guarantee control of FDRc.  Nevertheless, to the extent that the test used to combine the $p$-values of previously-undetected nodes is powerful when non-null effects appear in most or all individual tests, we can expect controlling FAR should be similar to controlling FDRc. We return to this issue in section 2.3, where this reasoning leads us to advocate use of Stouffers Z-score over Fisher's method, when testing the modified null hypothesis.

The modified null hypothesis we test has another important implication:  a test at one level may decide hypotheses at one or more higher levels. This occurs when a node has no undetected offspring, i.e. $\mathcal{U}_{l,j}=\emptyset$.  For example, since both offspring of $N_{2,6}$ in Figure \ref{tree}(a) are detected, we should immediately conclude that $N_{2,6}$ is associated.  Similarly, having already detected $N_{2,6}$, if we determine that $N_{2,5}$ is associated, then $N_{3,3}$ would have no undetected offspring and should be determined to be associated as well.  We present two approaches to account for the effect that this multiplicity has on the FAR. In the first approach, we do not allow this propagation, and instead consider that the test of the last undetected offspring of a parent node is in fact a test of the parent node (or, in general, a test of the highest node decided by this single test).  So, for example, in Figure \ref{tree}(a), if $N_{2,6}$ had already been detected, rejecting the modified null at $N_{2,5}$ would add $N_{3,3}$ to the list of detected nodes, but not $N_{2,5}$.  In this way, each hypothesis we test results in a single addition to the list of detected nodes.  Although this solution is unsatisfactory in many ways, it leads to a simpler procedure that serves as a useful intermediate result in developing our recommended approach (which we present in Section 2.4).

\bigskip

\noindent{\large\bf 2.3 Testing to Control FAR:  Simple (Unweighted) Proposal}

We now construct a testing procedure that tests the nodes in the tree level by level, starting at level $l = 1$.  For each level $l = 1, \ldots, L$, our testing procedures consist of two elements:  a set of thresholds to determine which nodes are detected at level $l$, and a way of aggregating the $p$-values from the undetected nodes at level $l$ to give $p$-values for the (parent) nodes at level $l+1$ for those nodes at level $l+1$ that have undetected offspring.  Our goal is to control the error rate so that the FAR $\leq q$.  In analogy with the concept of alpha spending in interim analysis \citep{demets1994interim}, we allocate to each level $l$ a target level $q_{l} \;(l = 1, 2,\ldots, L)$ chosen so that $\sum_{l=1}^L q_{l} = q$. We note here that we do not guarantee the FAR \textit{at each level} is controlled at level $q_l$, just that the \textit{overall} FAR is controlled at level $q$.  Although it would be interesting to develop an optimal strategy for choosing the $q_l$s, we choose $q_l = q n_l/n $, where $n_l$ is the number of nodes at level $l$ and $n = \sum_{l=1}^L n_l$ is the number of nodes in the tree.

We first consider how to assign $p$-value thresholds to control FAR.  Recall that for $l > 1$, tests at a lower level may have already resulted in detection of some of the nodes at level $l$.  Suppose that, of the $n_l$ total nodes at level $l$, there are $n_{l}^{*}$ nodes that have at least one child node that has not been detected.  Without loss of generality, assume the $p$-values for each node at level $l$,  $p_{l,1}, p_{l,2}, \ldots, p_{l,n_l^*}$, have been sorted in ascending order, and let the sorted values be denoted by $p_{l,(1)}\le p_{l,(2)}\le \cdots\le p_{l,(n_{l}^{*})}$. Let $d^*_{l}$ denote the (as yet unknown) number of nodes detected at level $l$.  We seek a set of ascending thresholds $\alpha_{l,1} \le \alpha_{l,2} \le \cdots \le \alpha_{l,n_{l}^{*}}$ by which we reject the modified null hypothesis at $d^*_l > 0$ nodes (corresponding to $p_{l,(1)},\ldots, p_{l,(d^*_l)}$) if $p_{l,(1)} \le \alpha_{l,1}, p_{l,(2)} \le \alpha_{l,2},\ldots, p_{l,(d^*_l)} \le \alpha_{l,d^*_{l}}$ but $p_{l,(d^*_{l}+1)}  > \alpha_{l,d^*_l+1}$; we accept the modified null hypothesis at all nodes in level $l$ if $p_{l,(1)} > \alpha_{l,1} $ in which case we take $d^*_{l}=0$, or reject the modified null hypotheses at all nodes in level $l$ if $p_{l,(1)} \le \alpha_{l,1}, p_{l,(2)} \le \alpha_{l,2},\ldots, p_{l,(n_{l}^*)} \le \alpha_{l,n_{l}^{*}}$ in which case we take $d^*_l=n_l^*$. We adopt the thresholds $\{\alpha_{l,j}\}$ given by 
\begin{eqnarray}
\label{threshold1}
\frac{\alpha_{l,j}}{1-\alpha_{l,j}} = \left(\frac{D_{l-1}+j}{n_l^* - j + 1} \times q_l \right)\bigwedge \frac{\tau_0}{1-\tau_0},
\end{eqnarray}
where $D_{l-1}=\sum_{l^{\prime}=1}^{l-1} d^*_{l^{\prime}}$ for $l\ge 2$ is the cumulative number of detection made up to and including the $(l-1)$th level, $D_{0}=0$, and $\tau_0$ is a pre-specified constant to prevent nodes with large $p$-values from being detected if a large number (say, $m$) of null hypotheses can be easily rejected because of very low $p$-values, in which case $q \times m$ nodes with large $p$-values can be said to be detected, while still controlling the overall error rate at level $q$; we set $\tau_0 = 0.3$ in this article. At each level, the thresholds \eqref{threshold1} are a variant of the thresholds in the step-down test proposed by \cite{gavrilov_adaptive_2009}, which have been used to control FDR in some applications as they have been shown to be more powerful than the standard BH procedure. Theorem 1 asserts that our bottom-up procedure with thresholds \eqref{threshold1} control the FAR at $ q$.
\begin{theorem}{1}
Assume the three conditions hold: (C1) nodes on the same level have the same depth; (C2) $p$-values for null nodes follow the uniform distribution $\mathrm{U}[0,1]$; (C3) at each level, the $p$-value for a null node is independent of the $p$-values at all other nodes. Then the bottom-up procedure with thresholds \eqref{threshold1} ensures that the FAR $\leq q$.
\end{theorem}
The proof of this theorem is provided in Appendix A.1. Condition (C1) assumes that nodes on the same level have the same depth, and will be relaxed in Section 2.5. Conditions (C2) and (C3) can be satisfied by our proposal below for obtaining $p$-values for parent nodes.

We now consider how to aggregate the $p$-values from level $l$ that correspond to the undetected offspring of a node at level $l+1$. Note that each undetected node at level $l$ is a member of exactly one of the sets $\mathcal{U}_{l+1,j}$, $j=1,\ldots,n^*_{l+1}$, the collections of the undetected offspring of nodes at level $l+1$. Thus, for the $j^\text{th}$ node at level $l+1$, we pool information from nodes in $\mathcal{U}_{l+1,j}$.  Note that the $p$-values of the undetected nodes at level $l$ necessarily exceed the threshold $\alpha_{l,d^*_l+1}$, and are hence not uniformly distributed on the interval $[0,1]$.  However, since this is the only restriction on these $p$-values,  it follows that, under either the global or modified null hypothesis, the $p$-values for nodes that were not detected at level $l$ are uniformly distributed on the interval $[\alpha_{l,d^*_l+1},1]$;  equivalently, adjusted $p$-values $p_{l,k}^{\prime}=(p_{l,k} - \alpha_{l,d^*_l+1})/(1-\alpha_{l,d^*_l+1})$ are uniformly distributed on the interval $[0,1]$.   Thus, we form Stouffer's Z score:
\begin{equation}
\label{eq3}
Z_{l+1,j} 
= \frac{1}{\sqrt{|\mathcal{U}_{l+1,j}|}}\sum_{k\in \mathcal{U}_{l+1,j}} \Phi^{-1}(1 - p_{l,k}^{\prime}), 
\end{equation}
where $\Phi$ is the standard normal cumulative distribution function and $|\mathcal{U}_{l+1,j}|$ represents the cardinality of $\mathcal{U}_{l+1,j}$. The $Z_{l+1,j}$ calculated using the undetected null nodes in $\mathcal{U}_{l+1,j}$ follows a standard normal distribution $N(0,1)$ under the modified null conditional on the pattern of detection at level $l$. Thus, $p_{l+1,j}=1 - \Phi(Z_{l+1,j})$;  in addition, $p_{l+1,j}$ and $p_{l+1,j^{\prime}}$ are independent since, on any tree, $\mathcal{U}_{l+1,j} \cap \mathcal{U}_{l+1,j^{\prime}}=\emptyset$. 

We use Stouffer's Z-score as it is known to be powerful when small or moderate non-null effects appear in the majority of individual tests as opposed to Fisher's method, which is additionally powerful when only a few large non-null effects are present \citep{loughin2004systematic}. As a result, using Stouffer's Z-score when controlling the FAR gives a better control of FDRc than using Fisher's method (the results based on Fisher's method not shown).

To summarize our procedure, we start at level $l=1$ with the $p$-values at the leaf nodes (which are given).  We determine which are detected and which are not detected using thresholds $\alpha_{1,j}$ calculated using \eqref{threshold1}.  For any nodes at level $l=2$ that have undetected offspring (i.e., for which $\mathcal{U}_{2,j} \neq \emptyset$), we then aggregate these undetected $p$-values into a $Z$ score using \eqref{eq3} and convert the value of this statistic into a $p$-value, which then serves as the $p$-value for the (parent) nodes on level $l=2$.  We continue in this manner until we reach the root node of the tree.  For this simplified approach, our list of detected nodes consists of each node that was detected at level $l$ for nodes that did not result in multiple detection for a single test, or the highest node detected for those nodes whose detection resulted in nodes at a higher level also being detected.

\bigskip

\noindent{\large\bf 2.4 Testing to Control FAR:  Full (Weighted) Procedure}

The testing procedure described in Section 2.3 may be unsatisfactory to many users because, when a test at a single node results in detection of multiple nodes, only one can be included on the list of discoveries if we wish to control FAR.  Thus, we consider a modification of this algorithm which allows all the detected nodes to be considered discoveries.

The difficulty with including all discoveries made by the approach in Section 2.3 is that an incorrect decision on a single node can result in multiple false discoveries.  To resolve this, we introduce weights $\omega_{l,j}$ that count the number of detections that could arise when testing node $N_{l,j}$ under the modified null hypothesis (i.e., conditional on the detections at lower levels). Consider the example shown in Figure \ref{tree}(a).  Assume at level 1, nodes $N_{1,11}$ and $N_{1,12}$ have been determined to be significant. Then, at level 2, the remaining nodes to be tested are $N_{2,1}$--$N_{2,5}$ (inside the black box). We consider testing nodes at each level in ascending order of $p$-values and assume $p_{2,1} < p_{2,2} < p_{2,3} < p_{2,4} < p_{2,5}$.  Rejecting the modified null hypothesis at $N_{2,1}$ only detects $N_{2,1}$; rejecting the null at $N_{2,2}$ will detect $N_{2,2}$ and $N_{3,1}$ (2 nodes); rejecting the null at $N_{2,3}$ will detect $N_{2,3}$; rejecting the null at $N_{2,4}$ will detect $N_{2,4}$ and $N_{3,2}$ (2 nodes); rejecting the null at $N_{2,5}$ will detect $N_{2,5}$, $N_{3,3}$, and $N_{4,1}$ (3 nodes).  Thus, for this ordering, we define the weights $(\omega_{2,1}, \omega_{2,2}, \omega_{2,3}, \omega_{2,4}, \omega_{2,5})=(1,2,1,2,3)$. Now, suppose instead that the $p$-values were ordered as $p_{2,5} < p_{2,1} < p_{2,2} < p_{2,3} < p_{2,4}$; then the weights would be $(\omega_{2,5}, \omega_{2,1}, \omega_{2,2}, \omega_{2,3}, \omega_{2,4})=(2,1,2,1,3)$.  Although these weights are different, the \textit{sorted} weights $(\omega_{2,(1)}, \omega_{2,(2)}, \omega_{2,(3)}, \omega_{2,(4)}, \omega_{2,(5)}) = (1, 1, 2, 2, 3)$ are the same. It is easy to verify that the same set of \textit{sorted weights} will be obtained with any other ordering of nodes.  Thus, the (unsorted) weights $(\omega_{l,1}, \omega_{l,2}, \ldots, \omega_{l, n_l^*})$ depend on the $p$-values at level $l$ and are thus random (even conditional on the detection events below level $l$). However, for complete trees (i.e., under Condition C1), we show in Appendix A.2 that the sorted weights $\omega_{l,(1)}\le \omega_{l,(2)}\le\cdots \le \omega_{l, (n_l^*)}$ are unique regardless of the ordering of $p$-values.  

Using the weights just defined, the FAR we wish to control becomes 
\begin{equation}
\label{eq1.2}
\mathrm{FAR} = \mathbb{E}\left[\frac{\sum_{l=1}^{L}  \sum_{j=1}^{n_l^*} \omega_{l,j}V_{l,j}^{m}}{\left(\sum_{l=1}^{L} \sum_{j=1}^{n_l^*} \omega_{l,j}R_{l,j} \right)\bigvee 1}\right].
\end{equation}
We modify the thresholds in \eqref{threshold1} by replacing the count $j$ with $\sum_{k=1}^j\omega_{l,(k)}$ and $n_l^*-j+1$ with $\sum_{k=j}^{n_l^*}\omega_{l,(k)}$ to give:
\begin{eqnarray}
\label{threshold2}
\frac{\alpha_{l,j}}{1-\alpha_{l,j}} = \left(\frac{D_{l-1}+\sum_{k=1}^j\omega_{l,(k)}}{\sum_{k=j}^{n_l^*}\omega_{l,(k)}} \times q_l \right)\bigwedge \frac{\tau_0}{1-\tau_0}.
\end{eqnarray}
Theorem 2 (proved in Appendix A.2) asserts that the bottom-up procedure with thresholds \eqref{threshold2} controls FAR \eqref{eq1.2} at value $ \leq q$.
\begin{theorem}{2}
Under Conditions (C1), (C2) and (C3) in Theorem 1, the bottom-up procedure with thresholds \eqref{threshold2} ensures that the value of the FAR given in equation \eqref{eq1.2} is $\leq q$.
\end{theorem}

\bigskip

\noindent{\large\bf 2.5 Bottom-up Testing on Incomplete Trees}

In Sections 2.3--2.4, we only considered complete trees where nodes on the same level all have the same depth. Here we consider more general trees where depth and level do not coincide. For example, in the tree from Figure \ref{tree}(b), nodes $N_{1,3}$ and $N_{1,4}$ have different depth from the other leaf nodes, although they are all on the same level.  In the microbiome example, this would occur whenever some of the lower taxonomic ranks (e.g., species and genus) of an OTU are not known.  One possible solution is to fill in the missing levels by assigning each such OTU its own (unknown) species and (unknown) genus. This is unsatisfactory both scientifically, as we then assert that the ``correct" species and genus for this OTU is different from any other OTU, and statistically, as the $p$-value for the species and genus level tests are necessarily identical to the $p$-value for the OTU. Although this may seem similar to the situation addressed in Section (2.4) where a single test could determine multiple hypotheses, it is actually different in that there is no additional information gained at the species or genus level in this example. A better alternative is to place each leaf node at the level just below the nearest inner nodes. However, this strategy can still be unsatisfactory for applications such as the microbiome, as it is scientifically questionable to treat some OTUs at the same level as higher taxa such as families. As a result, we describe here how our approach can be extended to incomplete trees.

For an incomplete tree, the sorted weights $(\omega_{l,(1)}, \omega_{l,(2)}, \ldots, \omega_{l,(n_l^*)})$ are no longer unique. In Figure \ref{tree}(b), when $N_{1,6}$ has the largest $p$-value among all leaf nodes, the sorted weights at level 1 are $(1,1,1,1,2,3)$; if $N_{1,4}$ has the largest $p$-value, the sorted weights are $(1,1,1,2,2,2)$. To account for this ambiguity, we seek a single set of sorted weights that will control FAR for any possible ordering of \textit{p}-values.  For the two sets of weights just considered, note that the cumulative sums of sorted weights $\sum_{k=1}^j\omega_{l,(k)}$ for the first set, given by $(1,2,3,4,6,9)$ are all less than or equal to the cumulative sums of the sorted weights of the second set, given by $(1,2,3,5,7,9)$.  Thus, if we were to use the first set of ordered weights in (4), the thresholds $\alpha_{l,j}$ would be smaller than the thresholds calculated using the second set of sorted weights.  In Appendix A.3, we show how to find a unique set of sorted weights $\widetilde{\omega}_{l,(1)}\le\cdots\le\widetilde{\omega}_{l,(n_l^*)}$ for level $l$ that correspond to weights obtained by some ordering of $p$-values and that satisfy the inequalities 
$$
\sum_{k=1}^j\widetilde{\omega}_{l,(k)} \le \sum_{k=1}^j\omega_{l,(k)}, ~~j=1,\ldots,n_l^*,
$$ 
for all possible sets of sorted weights $(\omega_{l,(1)},\ldots,\omega_{l,(n_l^*)})$ induced by different orderings of $p$-values. We then adopt thresholds calculated using $\{\widetilde{\omega}_{l,(k)}\}$, given by
\begin{eqnarray}
\label{threshold3}
\frac{\alpha_{l,j}}{1-\alpha_{l,j}} = \left(\frac{D_{l-1}+\sum_{k=1}^j \widetilde{\omega}_{l,(k)}}{\sum_{k=j}^{n_l^*}\widetilde{\omega}_{l,(k)}} \times q_l \right)\bigwedge \frac{\tau_0}{1-\tau_0}.
\end{eqnarray}
Because these thresholds are the most stringent among those based on any possible weights $(\omega_{l,(1)},\ldots,\omega_{l,(n_l^*)})$, we call them the least favorable weights. Theorem 3 ensures control of the FAR using the least favorable weights. 
\begin{theorem}{3}
Under Conditions (C2) and (C3) in Theorem 1, the bottom-up procedure with thresholds \eqref{threshold3} ensures the FAR defined in \eqref{eq1.2} is $\leq q$.
\end{theorem}
The proof of Theorem 3 can be found in Appendix A.4.  When a tree is complete, the least favorable weights reduce to the unique sorted weights regardless of the ordering of $p$-values. Thus the testing procedure presented here encompasses the one presented in Section 2.4 as a special case.

\bigskip

\noindent{\large\bf 2.6 Bottom-up Testing with Separate FAR Control}

The testing procedures we have described so far assume that we wish to detect nodes at all levels of the tree while controlling the overall FAR at some level $q$.  In some situations we may want to conduct a separate analysis of leaf nodes and inner nodes.  For example, we may wish to first determine which OTUs are detected while controlling FAR at some level $q_1$; then we may wish to conduct a second, separate analysis of taxa starting at the species level and continuing up the phylogenetic tree, while controlling the FAR of the second analysis at some level $q_{-1}$.

The procedures presented in Sections 2.4 and 2.5 do not guarantee that the FAR at each level $l$ is controlled at level $q_l$ because of the cumulative effect of $D_{l-1}$ in \eqref{threshold1}, \eqref{threshold2} and \eqref{threshold3}, which establishes a dependence between the nodes detected at each level.  If we break this dependence by re-starting the counter at some level, it is then possible to separately control FAR above and below this level. Here, we illustrate our proposal by showing how to control FAR at level 1 to have value $\leq q_1$ while simultaneously controlling FAR at all remaining (higher) levels at value $\leq q_{-1}$. 
To accomplish this, we propose a two-stage procedure. At stage 1, we perform the step-down test for level 1 with thresholds $\{\alpha_{1,j}, j = 1,2,\ldots,n_1\}$ that satisfy 
$$
\frac{ \alpha_{1,j}}{1-\alpha_{1,j}} = \left( \frac{\sum_{k=1}^j \widetilde{\omega}_{1,(k)}}{\sum_{k=j}^{n_1}\widetilde{\omega}_{1,(k)}}\times q_1 \right)\bigwedge \frac{\tau_0}{1-\tau_0}.
$$
Note that, if the same value of $q_1$ is used, these are the same thresholds for level $1$ as the one-stage procedure described in the previous section.  The use of weights $\{\widetilde{\omega}_{1,(k)}\}$ to account for multiplicity allows us to include in our list of detected nodes at level 1 all higher-level nodes that are detected after testing level 1 nodes.  Thus, the FAR at level 1 is written as 
$$
\mathrm{FAR_{otu}} 
= \mathbb{E} \left[\frac{\sum_{j=1}^{n_1} {\omega}_{1,j}V_{1,j}^{m}}{\left(\sum_{j=1}^{n_1} {\omega}_{1,j}R_{1,j}\right)\bigvee 1}\right],
$$

At stage 2, we then apply the one-stage procedure proposed in Section 2.4 or 2.5 to the tree obtained by removing all leaves (OTUs) as well as those higher-level taxa that were detected at stage 1. In this tree, undetected nodes at level 2 are now the leaves, and the $p$-values for these new leaves are calculated by aggregating the $p$-values from the undetected OTUs.  We use the thresholds $\{\alpha_{l,j}, l=2,\ldots,L, j=1,\ldots,n_l^*\}$ that satisfy 
$$
\frac{ \alpha_{l,j}}{1-\alpha_{l,j}} = \left( \frac{D^{\dagger}_{l-1}+\sum_{k=1}^j \widetilde{\omega}_{l,(k)}}{\sum_{k=j}^{n_l^*}\widetilde{\omega}_{l,(k)}}\times q_l\right)\bigwedge \frac{\tau_0}{1-\tau_0},
$$
where $D^{\dagger}_{l-1}=\sum_{l^{\prime}=2}^l d^*_{l^{\prime}}$ differs from $D_{l-1}$ in that $D^{\dagger}_{l-1}$ counts the detected nodes starting from the $2^\text{nd}$ level. Using $D^{\dagger}_{l-1}$ in place of $D_{l-1}$ cuts the dependence between level 1 and the remaining levels and also makes the thresholds more stringent if $q_l$s stay the same as those in the one-stage procedure.  Thus, the FAR we wish to control at the remaining (higher) levels is 
$$
\mathrm{FAR_{taxa}} 
=  \mathbb{E}\left[ \frac{\sum_{l=2}^{L}  \sum_{j=1}^{n_l^*} {\omega}_{l,j}V_{l,j}^{m}}{\left(\sum_{l=2}^{L} \sum_{j=1}^{n_l^*} {\omega}_{l,j}R_{l,j}\right) \bigvee 1}\right],
$$
where, as before, $n_l^*$ excludes nodes detected using information from level $1$ to $l-1$.  Theorem 4 states that this two-stage procedure serves our purpose.
\begin{theorem}{4}
Under Conditions (C2) and (C3) in Theorem 1, the above two-stage procedure ensures that $\mathrm{FAR_{otu}} \leq q_1$ and $\mathrm{FAR_{taxa}} \leq q_{-1}= \sum_{l=2}^L q_l$.
\end{theorem}
The proof of Theorem 4 is proved in Appendix A.5.  Note that the choice of $(q_1, q_2, \ldots, q_L)$ is at the user's discretion and not necessary to match those in the one-stage procedure. For example, we can set $q_1=q_{-1}=5\%$ and choose $q_l=q_{-1}n_l/\left(\sum_{l'=2}^L n_l\right)$ for $l=2,\ldots,L$. 

Although we have presented the example in which the FAR for leaf nodes are controlled separately from the FAR for inner nodes, in principal the approach described here could be used to divide the nodes into two groups at any level, simply by choosing where to zero the counter in $D_l$.  We could even apply the approach recursively to control FAR for more than two sets of levels, if desired.

\section{Simulation Studies}
\label{sec:simu}

We conducted simulation studies to assess the performance of our bottom-up tests, and to compare with three competing approaches: (1) the na\"{\i}ve approach that calculates the $p$-value for an inner node by aggregating $p$-values from \textit{all} leaf nodes that are its descendents, using Stouffer's Z-score method and applies the BH procedure on the collection of $p$-values from all nodes; (2) the top-down approach of \cite{yekutieli_hierarchical_2008} as implemented in the R package \texttt{structSSI} \citep{sankaran_structssi:_2014}, with $p$-values for inner nodes calculated in the same way as in the na\"{\i}ve approach; and (3) the conjunction-null test that assigns a $p$-value to an inner node by the \textit{largest} $p$-value from all offspring nodes (equivalently, the largest $p$-value from all corresponding leaf nodes) and applies the BH procedure as in the na\"{\i}ve approach. All methods take $p$-values at leaf nodes as input. The nominal level for all error rates was set to 10\%.

All simulations were conducted under the conjunction null.  We first selected a number of inner or leaf nodes to be driver nodes. Under the conjunction null, all offspring of driver nodes (including all its leaf nodes) are associated with the trait of interest. We independently sampled $p$-values for associated leaf nodes from distributions that have enriched probability at values close to zero. We used the Beta distribution $\text{Beta}(1/\beta,1)$ where $\beta>1$, which has a relatively heavy right tail (Figure \ref{simulated_pvalue}) and mimics the empirical distributions of $p$-values observed in the IBD data (Figure \ref{empirical_pvalue}). To assess the robustness of our results, we also considered sampling $p$-values from a Gaussian-tailed model frequently used to study the performance of FDR procedures \citep{storey2002direct, barber2017p, javanmard2018online}, which has a smaller right tail (Figure \ref{simulated_pvalue}). Specifically, we first drew values $X_{l,j}\sim N (\beta, 1)$ and then obtained the $p$-value $p_{l,j}=1 - \Phi(X_{l,j})$, where $\Phi$ is the standard normal cumulative distribution function. In both models, $\beta$ characterizes the effect size of the trait on the microbiome. For all simulations we assumed that all leaf nodes that were not descendants of driver nodes were null with $p$-values sampled independently from the $\mathrm{U}[0,1]$ distribution.

\begin{figure}
\centering
\includegraphics[width=1.0\textwidth]{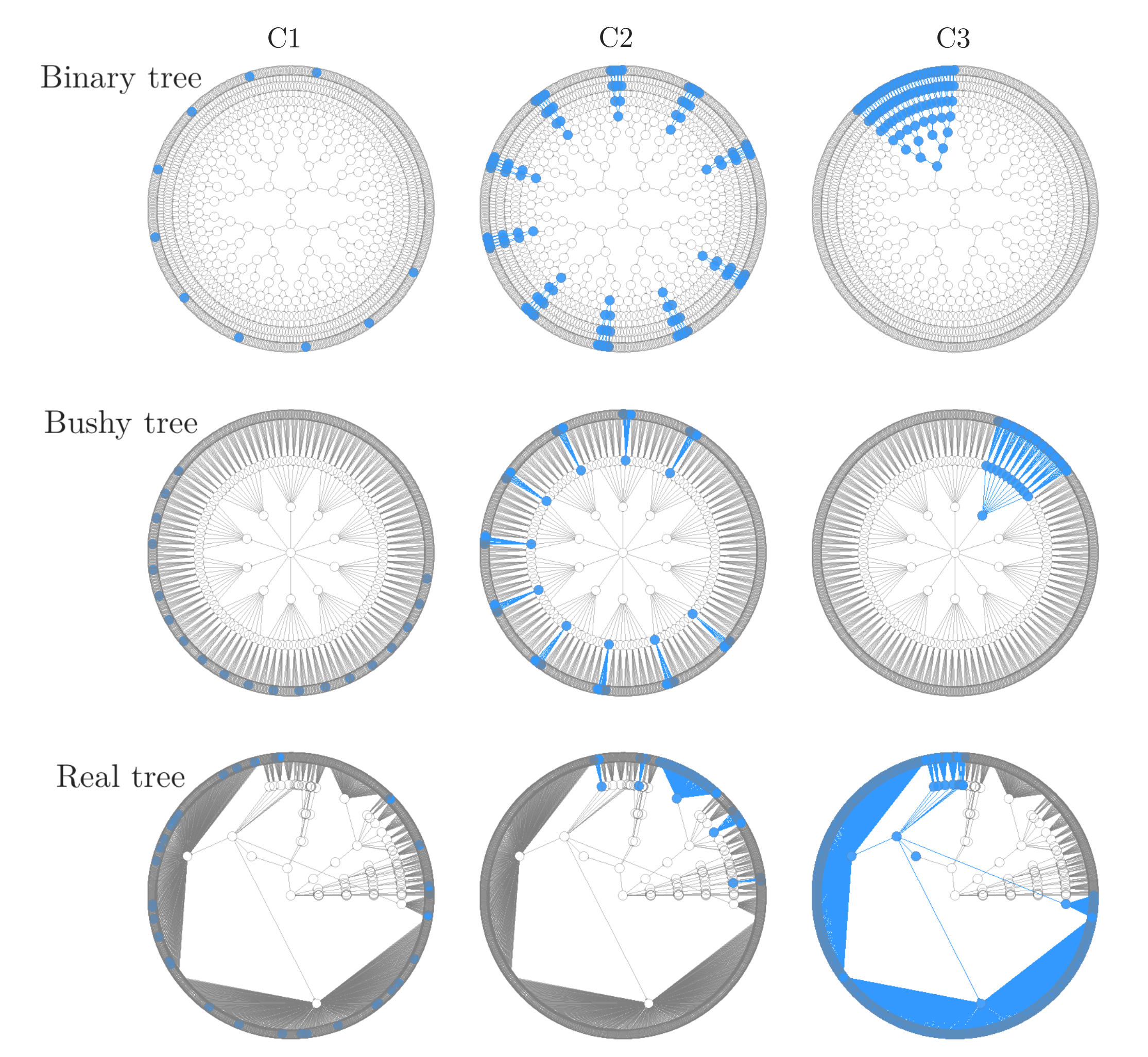}
\caption{\label{fig:causality}{The three tree structures (binary, bushy, and real) and three causal patterns  (C1, C2, and C3) for simulation studies. The real phylogenetic tree structure was obtained from the IBD data and, for simplicity of exposition, skipped the genus and species levels which have extensive missing assignments. The root node is always located at the center of each tree and the leaf nodes are represented by the outermost ring. The top of each blue subtree is a designated driver node, which can be an inner or leaf node.}}
\end{figure}

We considered three tree structures, shown in Figure \ref{fig:causality}. The first is a complete binary tree with 2 children for each inner node and 10 levels, which has 1023 nodes of which 512 are leaves. The second is a complete `bushy' tree with 10 children for each inner node and 4 levels, which has 1111 nodes of which 1000 are leaves.  The third is a real phylogenetic tree \citep{halfvarson_dynamics_2017} that we apply our methods to in Section 4. This tree has 8 levels, 249 inner nodes, and 2360 leaf nodes, with large variation in the number of child nodes at different inner nodes. It is incomplete, having extensive ($>$ 50\%) missing assignments at the genus and species levels, and a few at the family level.

For each tree structure, we then considered three causal patterns, differentiated by the level of the selected driver nodes (Figure \ref{fig:causality}). The first pattern (C1) is characterized by sparse driver nodes located at the leaf nodes; the second (C2) by several driver nodes located at an intermediate level, chosen so that $\sim$10\% of leaf nodes were associated; and the third (C3) by a single driver node at a higher level, inducing association with a large subtree.  In particular, for C1 we randomly selected 10/20/36 leaf nodes for the binary/bushy/real trees. For C2, we randomly chosen 10 out of 64 nodes at level 4 for the binary tree, 10 out of 100 nodes at level 2 for the bushy tree, and 5 out of 48 nodes at the family level (level 4) for the real tree. For C3, we randomly picked one node at level 7 for the binary tree, one node at level 3 for the bushy tree, and for the real tree, the class \textit{Clostridia} from level 6 (covering $\sim$80\% of all leaf nodes).  Each of the $3 \times 3=9$ scenarios was replicated 1000 times.

\bigskip

\noindent{\large\bf 3.1 Error Rates}

We evaluated each procedure by calculating its FAR (under the modified null), FDR (under the global null), and its FDRc (under the conjunction null). The FAR of the top-down and na\"{\i}ve approaches were calculated by using the list of detected nodes (as determined by \cite{benjamini_controlling_1995} for the na\"{\i}ve method and \cite{yekutieli_hierarchical_2008} for the top-down method), then proceeding level by level as if these detections were the result of tests of the modified null hypothesis.  The FAR was then calculated using the weighted procedures of Section 2.4 for the binary and bushy trees and Section 2.5 for the real tree.

Figure \ref{error_rates} displays these results for the nine ($= 3 \times 3$) scenarios we considered, for the simulations that used the beta distribution for non-null leaves. In all 9 scenarios our methods, whether unweighted or weighted, always controlled FAR. In contrast, the na\"{\i}ve and top-down methods typically had inflated FAR, with the most severe inflation occurring in C1, where the driver nodes were simulated exclusively at leaf nodes, causing many higher-level nodes to be falsely detected by these methods. As expected, our methods also controlled FDR. The top-down method, although designed to control FDR, still yielded slightly inflated FDR occasionally; this is likely due to violation of the independence assumption between the $p$-value at a node and the $p$-values of all its ancestors, which is required by the top-down method \citep{yekutieli_hierarchical_2008}. The na\"{\i}ve method always controlled FDR because the BH procedure is known to be robust to such positive correlations. Despite a lack of theoretical results, we found that our methods (especially the weighted approach of Section 2.4) controlled FDRc reasonably well in all scenarios we considered. The na\"{\i}ve and top-down methods typically had inflated FDRc, and FDRc for these methods resembled their FAR, consistent with the notion that FAR approximates FDRc. The conjunction-null test controlled all error rates, as expected. Figure \ref{error_rates_gaussian} shows the same patterns of FAR, FDR, and FDRc for simulations based on the Gaussian-tailed model.

\begin{figure}
\centering
\includegraphics[width=1\textwidth]{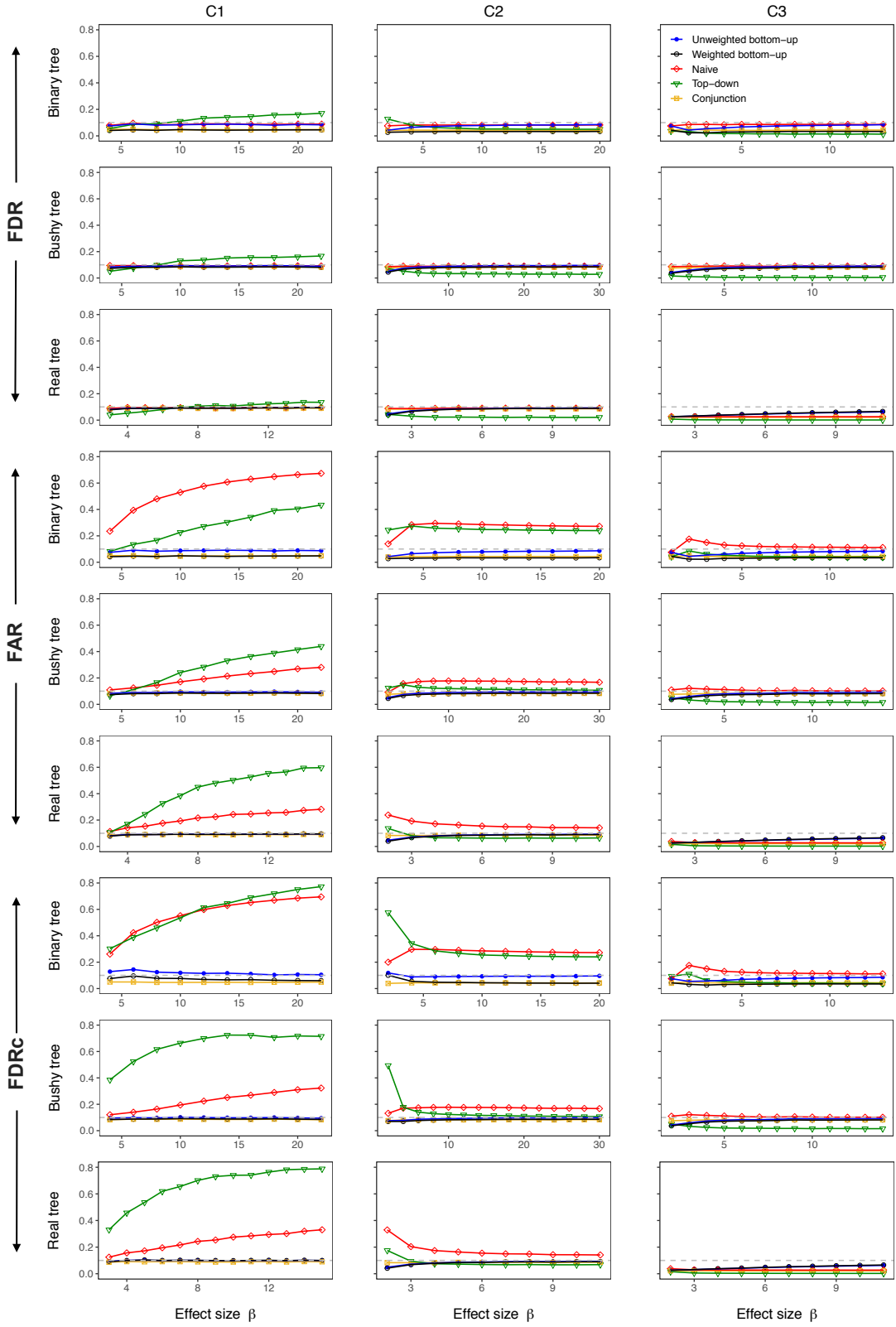}
\caption{\label{error_rates} Error rates for testing all nodes in the tree. The non-null $p$-values at leaf nodes were simulated from the beta distribution.}
\end{figure}

\bigskip

\noindent{\large\bf 3.2 Accuracy and Pinpointing Driver Nodes}

As our simulations were conducted under the conjunction null hypothesis, the driver nodes and all of its descendants are the truly associated nodes.  We measured accuracy by calculating a Jaccard similarity between the set of truly associated nodes and the nodes that are detected, for each method.  We used a \textit{weighted} Jaccard similarity to account for the branching-tree topology of the hypotheses we test, because detecting a node with offspring implies the offspring are detected in some sense, even if they were not individually detected. For example, identifying a genus as being associated with the trait of interest implies the species and OTUs that belong to this genus are associated, even if we did not detect them (and hence are not included in the list of detected nodes).  For this reason, we calculated the Jaccard similarity by weighting each node by the number of leaf nodes that are its descendants.  The weighted Jaccard similarity is then the sum of the weights of correctly-detected nodes, divided by the total weight assigned to either detected or truly associated nodes.

Examining Figure \ref{sensitivity} we see that the weighted bottom-up approach has the best or second-best accuracy (as measured by the weighted Jaccard similarity) in all cases we examined, making it the best overall choice.  The test of the conjunction hypothesis slightly outperforms our bottom-up tests when only leaf nodes (OTUs) are truly associated (simulation C1) as its natural conservatism is an advantage in this case.  As soon as inner nodes are truly associated, as in C2 or C3, the conjunction test becomes very conservative. The bottom-up approaches performed best for C2 and C3 except for the binary tree in C3, where the na\"{\i}ve and top-down approaches performed better at low parameter values.  When considering these results, it should also be noted the na\"{\i}ve method (but not the top-down method) had elevated FAR and FDRc for this simulation.  In general, the performance of our weighted bottom-up procedure was either equivalent to or superior to the unweighted procedure, presumably reflecting the ability of the weighted procedure to include all nodes that are identified when all offspring of some node are detected.

\begin{figure}
\centering
\includegraphics[width=0.8\textwidth]{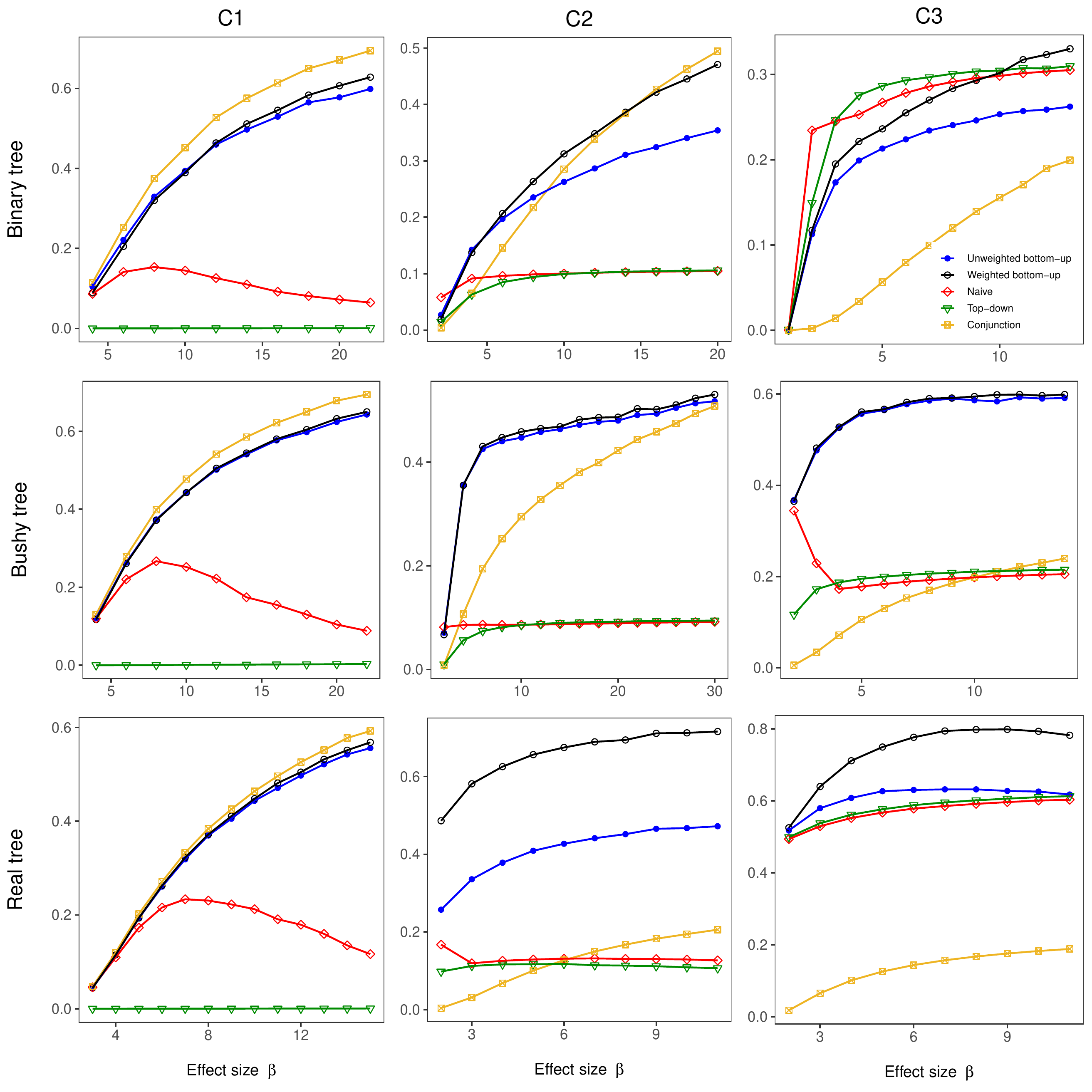}
\caption{\label{sensitivity} Accuracy (weighted Jaccard similarity) for detecting all associated nodes (including the driver nodes and all of their descendants at all levels). The non-null $p$-values at leaf nodes were simulated from the beta distribution.}
\end{figure}

Our methods are most different from existing methods in their ability to pinpoint driver nodes. We say a driver node is ``pinpointed" if it is detected to be associated \textit{and} none of its ancestors are detected. We evaluated the percentage of driver nodes that were pinpointed and showed the results in Figures \ref{driver_nodes} and \ref{driver_nodes_gaussian}. In general, our weighted and unweighted methods pinpointed a similar number of driver nodes, and both detected many more driver nodes than the na\"{\i}ve, top-down, and conjunction-null methods. The na\"{\i}ve method pinpointed some driver nodes when the association signals were weak, but inevitably detected their ancestors as the signals became stronger. By definition, the top-down method must fail to pinpoint any driver node since it only tests nodes below the root node if the root node is detected. Note that the percentage of driver nodes pinpointed by our methods sometimes decreased as the effect size increased, because more (but not all) descendants were detected and removed from the statistic for the driver nodes, which thus aggregated less information. For the beta-based simulations, undetected driver nodes remained a possibility regardless of the effect size, as the beta distribution always generates a non-negligible portion of $p$-values that were close to one even when the effect size was extremely large.  For the Gaussian-tailed simulation, all driver nodes were eventually detected, as large $p$-values became more infrequent as the effect size increased.  We note that the conjunction-null test can easily fail to detect higher-level nodes (including driver nodes) when some offspring of these nodes have large $p$-values, as in our beta-based simulations.  Additionally, we note that the good accuracy (as measured by the weighted Jaccard similarity) of our approaches is related to their ability to pinpoint driver nodes.

\begin{figure}
\centering
\includegraphics[width=0.8\textwidth]{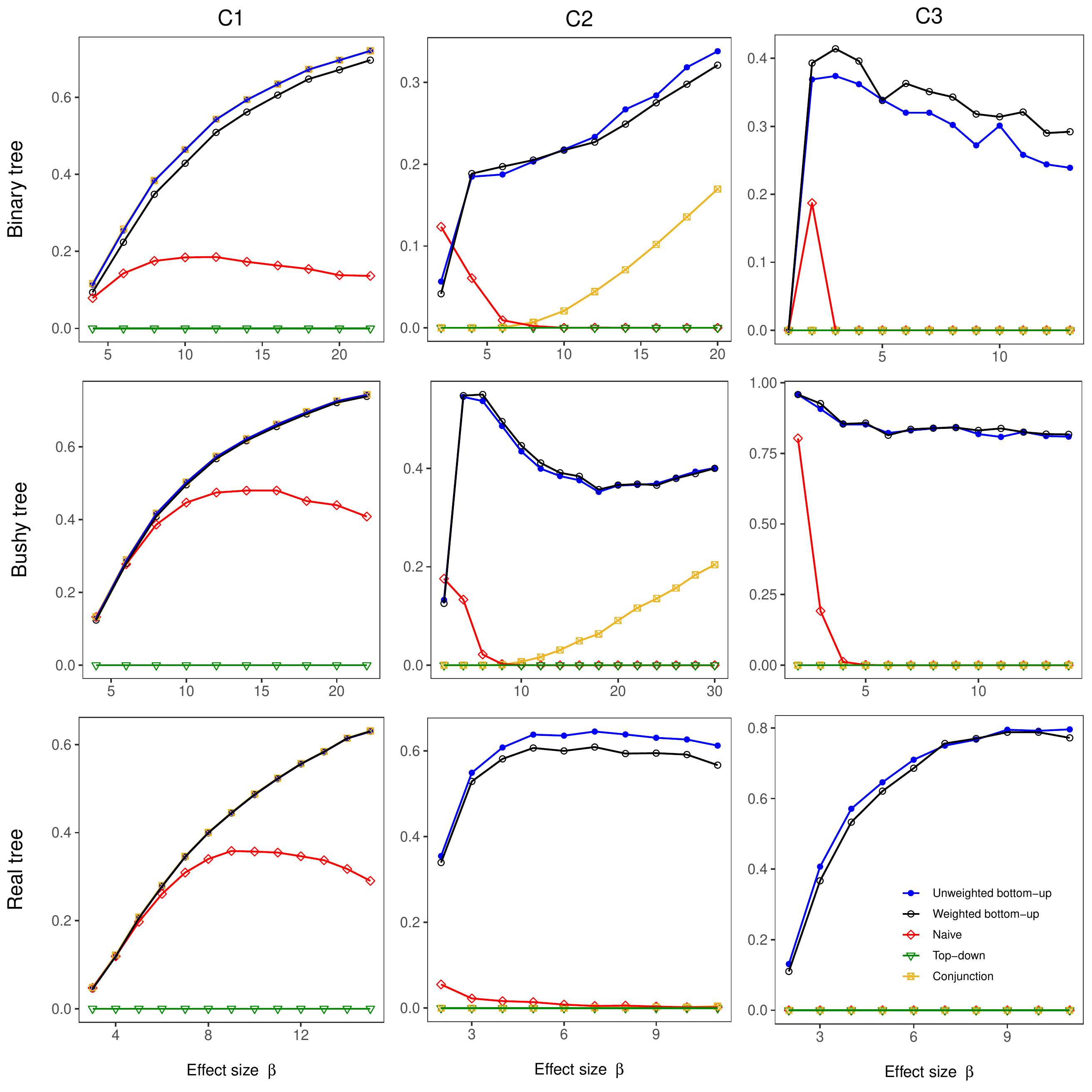}
\caption{\label{driver_nodes} Percentage of driver nodes that were pinpointed. The non-null $p$-values at leaf nodes were simulated from the beta distribution.}
\end{figure}

\section{IBD Data}
\label{sec:data1}
IBD is a chronic disease accompanied by inflammation in the human gut. The two most common subtypes are ulcerative colitis (UC) and Crohn's disease. \cite{halfvarson_dynamics_2017} investigated the longitudinal dynamics of the microbial community in an IBD cohort of 60 subjects with UC and 9 healthy controls. The microbial community was profiled by sequencing the V4 region of the 16S rRNA gene. Sequence data were processed into an OTU table through the QIIME pipeline. Our goal was to identify taxa that have differential abundance between the UC and control groups at baseline.

We removed OTUs that were present in fewer than 10 samples and dropped 4 OTUs that failed to be assigned any taxonomy. The assigned taxonomy grouped the 2360 OTUs into 249 taxonomic categories (i.e., inner nodes) corresponding to kingdom, phylum, class, order, family, genus, and species levels. Note that 15.2\%, 56.9\%, and 91.3\% OTUs have missing assignment at the family, genus, and species level, respectively. As there were no obvious confounders provided with these data, we used the Wilcoxon rank-sum test to compare the OTU frequencies between case and control groups to obtain $p$-values for each OTU.

\begin{figure}
\centering
\includegraphics[width=.865\textwidth]{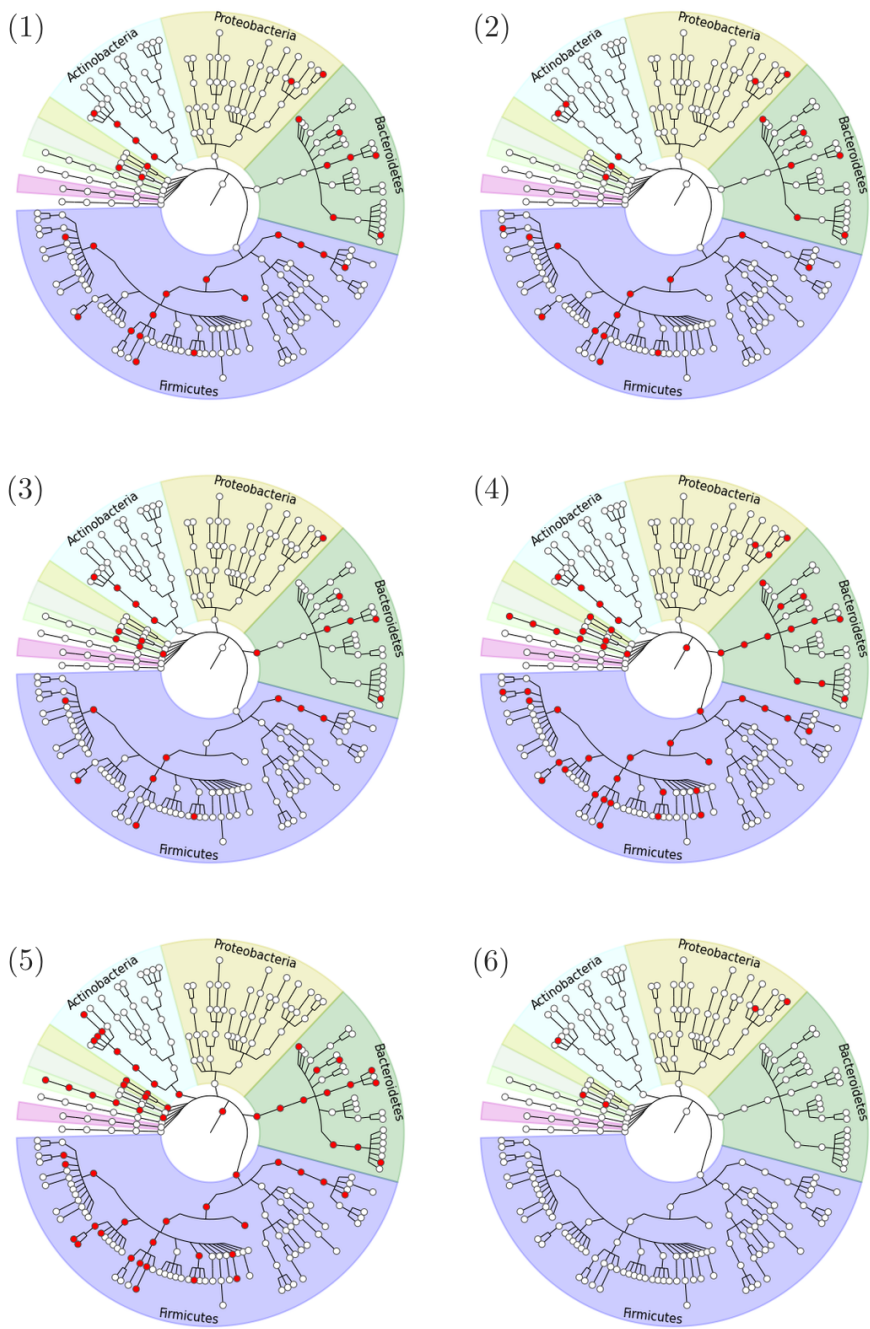}
\caption{\label{fig_UC_incomplete} Taxa (marked in red) detected to be differentially abundant between the UC and control groups in the IBD data by (1) the (one-stage) weighted bottom-up method, (2) the unweighted bottom-up method, (3) the two-stage, weighted bottom-up procedure with nominal FAR 5\% and 5\% for OTUs and taxa, respectively, (4) the na\"{\i}ve method, (5) the top-down method, and (6) the conjunction-null test. The levels from the center outward are kingdom, phylum, class, order, family, genus and species. The OTU level is supposed to be located at the outermost layer and has been omitted to simplify the figure. The plots were generated using GraPhlAn (\texttt{http://huttenhower.sph.harvard.edu/graphlan}).}
\end{figure}

We applied the two bottom-up methods with the nominal FAR of 10\% as well as the na\"{\i}ve and top-down methods with the nominal FDR of 10\%. The detected taxa can be visualized in Figure \ref{fig_UC_incomplete}. The weighted bottom-up test identified 143 OTUs, 5 species, 9 genera, 7 families, 5 orders, and 5 classes, among which the driver taxa were pinpointed at classes \textit{Chloroplast},  \textit{Clostridia},  \textit{Coriobacteriia}, \textit{Erysipelotrichi} and \textit{RF3}; families \textit{Bacteroidaceae}, \textit{Prevotellaceae} and \textit{S24-7};  genera \textit{Morganella} and  [\textit{Prevotella}]; and species \textit{radicincitans}; see Table \ref{tab:uchc} for more details. The unweighted procedure yielded an identical list of driver taxa.  In contrast, both the na\"{\i}ve and top-down methods identified the root node and many taxa at high levels, suggesting their inability to pinpoint the real driver taxa. In addition, the top-down method did not detect some lower-level taxa in phylum \textit{Proteobacteria} that were detected by all other methods. The conjunction test detected 142 OTUs but only 5 taxa, each of which contained a single OTU. All these results are consistent with the findings of our simulation studies.

We also applied the two-stage (weighted) bottom-up procedure to control FAR separately at the OTU level and the remaining taxa levels. We considered splitting the overall FAR 10\% into 5\% and 5\% for OTU and taxon analyses. At stage 1, we detected 81 OTUs and also included in our detection list 1 species, 1 genus, 1 order, and 1 class because they all contain only 1 OTU and those OTUs were detected; among these OTUs and taxa we can declare we control FAR at 5\% ($\sim$ 4 OTUs or taxa). At stage 2, we detected 4 species, 4 genera, 5 families, 4 orders, 3 classes, and 2 phyla, among which we can declare we control FAR at 5\% ($\sim$ 1 taxon). The detected driver taxa were phyla \textit{Bacteroidetes} and \textit{Verrucomicrobia}, classes \textit{Chloroplast}, \textit{Coriobacteriia} and \textit{Erysipelotrichi}, order \textit{Clostridiales}, and species \textit{radicincitans}. Note that the two FARs (for OTUs and taxa) do not have to add up to 10\% nor be equal.

\section{Discussion}
\label{sec:dis}
In this work, we presented a bottom-up approach to testing hypotheses that have a branching tree dependence structure. These procedures test hypotheses in a tree level by level, starting from the bottom and moving up, rather than starting from the top and moving down. We developed a novel modified null hypothesis, which is more suitable for our goal of detecting nodes in which a \textit{dense} set of child nodes are associated with the trait of interest. Accordingly, we developed a novel error criterion, the FAR, and provided procedures that we proved control FAR. Our simulation studies confirmed the control of FAR and demonstrated good performance of our methods compared to existing methods using a measure of accuracy based on a weighted Jaccard similarity. Further, our bottom-up methods are more successful at pinpointing driver nodes, offering highly interpretable results, while the existing methods frequently fail at this task. Finally, although our methods were not designed to control FDRc, our simulations showed that use of Stouffer's Z score to combine information leads to approximate control of FDRc as well.

Our methods can easily be extended to very general tree structures. We can easily handle trees in which the leaf nodes are not all at level 1. With some modifications, our methods can also be applied to trees with multiple (correlated) root nodes such as trees generated by pathways, by using our bottom-up testing procedure up to the level right below the root level and applying to the root level the standard BH procedure, which is robust to positive correlations. We expect this modified procedure to control FAR (and hence FDR and, approximately, FDRc).

Although our approach is very general, it does have some limitations or aspects that could benefit from further development.  First, we treated $p$-values at leaf nodes with equal weight. In some applications, different leaf nodes may have varying importance and may be weighted differently. Second, we partitioned the total error rate $q$ into $q_1,\ldots,q_L$ in proportion to the number of nodes at each level. It is unclear what the optimal partition would be. It is also of interest to consider alternative partitioning that can improve performance at pre-specified levels of particular importance if, for example, finding genera that were associated with a trait was of particular interest.  Further, we assumed independence between null leaf nodes because it is required by both Stouffer's method for combining $p$-values and the step-down procedure for controlling the error rate of decisions. It may also be of interest to extend our methods to account for correlations between leaf nodes, e.g., correlations between (null) OTUs.

Our methods have been implemented in the R package \texttt{BOUTH} (\underline{BO}ttom-\underline{U}p \underline{T}ree \underline{H}ypothesis testing), available on GitHub 
\href{https://github.com/yli1992/BOUTH}{links}.. Our program is computationally efficient as it only involves building the tree structure, calculating the thresholds, aggregating $p$-values for parent nodes, and sorting $p$-values. For example, for the one-stage weighted procedure on the IBD data, it took 1.3 seconds on a laptop with a 2.5 GHz Intel Core i7 processor and 8 GB RAM.

\section{Acknowledgements and Disclaimer}
This research was supported by the National Institutes of Health awards R01GM116065 (Hu). The findings and conclusions in this report are those of the authors and do not necessarily represent the official position of the Centers for Disease Control and Prevention.

\appendix 
\numberwithin{equation}{section}
\section*{Appendix}
\label{sec:latex}

\renewcommand{\theequation}{A\arabic{equation}}

\setcounter{equation}{0}

\noindent{\large\bf A.1 Proof of Theorem 1}

The following lemma is essentially the summation by parts formula and will be used in the proof of Theorem 1.
\begin{lemma}{1}
Suppose $\{a_1, \ldots,a_n\}$ and $\{b_1, \ldots,b_n\}$ are two sets of real numbers. Then,
$$
\sum_{k=1}^n a_kb_k = \sum_{k=1}^n (a_k-a_{k-1})B_k,
$$
where $B_k=\sum_{i=k}^n b_i$ and $a_0=0$.
\end{lemma}

\begin{proof}[\textit{Proof of Theorem 1}]
This proof is adapted from \cite{gavrilov_adaptive_2009} with modifications for our multi-level, bottom-up procedure. 
Let $R_l$ and $V_l$ denote the number of detection and false detection, respectively, under the modified null hypothesis at level $l$. Then the false assignment proportion (FAP) is written as
$$
\text{FAP}
=\frac{\sum_{l=1}^L V_l}{(\sum_{l=1}^L R_l)\bigvee 1} 
\le \sum_{l=1}^L\frac{V_l}{(\sum_{l'=1}^l R_{l'})\bigvee 1}
=\sum_{l=1}^L\frac{V_l}{(D_{l-1}+R_l)\bigvee 1}
=\sum_{l=1}^L\frac{V_l}{D_{l-1}+(R_l\bigvee 1)}.
$$

The key step to prove Theorem 1 is to show that for every level $l$
\begin{equation}
\label{p1:key_step}
\bE \left[ \frac{V_l}{D_{l-1}+(R_l \bigvee 1)} \bigg| \mathcal{G}_{l-1} \right] \le q_l,
\end{equation}
where $\mathcal{G}_{l-1}$ represents detection events below level $l$. The inequality \eqref{p1:key_step} does not guarantee the control of FAR at each level $l$ because of the cumulative effect of $D_{l-1}$, which establishes a dependence between the nodes detected at different levels. However, it leads to the control of overall FAR at $q$:
\begin{equation*}
\text{FAR}
=\bE(\text{FAP})
\le\sum_{l=1}^L \bE\left\{\bE\left[\frac{V_{l}}{D_{l-1}+(R_l\bigvee 1)}\bigg|\mathcal{G}_{l-1} \right]\right\}\le \sum_{l=1}^L q_l 
= q.
\end{equation*}

To prove (\ref{p1:key_step}), we omit the level index $l$ for simplicity of exposition. Thus we redefine the $l$th-level $p$-values to be $p_1, \ldots, p_{n^*}$, ordered $p$-values $p_{(1)} \le \cdots \le p_{(n^*)}$, and thresholds defined in (\ref{threshold1}) $\alpha_{1}\le \cdots \le\alpha_{n^*}$. We use $D_{-1}$ for $D_{l-1}$. We also omit $\mathcal{G}_{l-1}$ by acknowledging that the ensuing arguments are always conditional on the detection events below level $l$. Further, we denote the set $\left\{p_{(1)}\le \alpha_1,\ldots,p_{(k)}\le \alpha_k\right\}$ by $\mA_k$ ($k=1,\ldots, n^*$), which represents the case that the first $k$ ordered $p$-values are each below the first $k$ thresholds. 

To start, we use the definitions of $V_l$ and $R_l$ and rewrite the left hand side of (\ref{p1:key_step}) to be
\begin{equation}
\label{p1:rewrite}
\sum_{j\in\mH_0}\sum_{k=1}^{n^*}\bE
\left[\frac{\bI\left(\mA_k, p_{(k+1)} > \alpha_{k+1}, p_j \le \alpha_k\right) }{D_{-1}+k} \right],
\end{equation}
where $\mH_0$ denotes the set of modified null hypotheses at level $l$. Then, we replace the expectation by double expectations that first conditions on $p_j$ and apply Lemma 1 with $a_k=\bI\left(p_j \le \alpha_k\right)/(D_{-1}+k)$, $b_k=\Pr\left(\mA_k, p_{(k+1)} > \alpha_{k+1} \big| p_j\right)$, and $n=n^*$; note that $b_{n^*}=\Pr\left(\mA_{n^*} \big| p_j\right)$. Thus, \eqref{p1:rewrite} becomes
\begin{equation*}
\sum_{j \in \mH_0}\sum_{k=1}^{n^*} \bE 
\left\{\Pr\left(\mA_k \big| p_j\right)\times 
\left[\frac{\bI\left(p_j  \le \alpha_k\right)}{D_{-1}+k}-\frac{\bI\left(p_j \le \alpha_{k-1}\right)}{D_{-1}+k-1} \right]
\right\},
\end{equation*}
which can be reorganized as
\begin{equation}
\label{p1:A}
\sum_{j\in\mH_0}\sum_{k=1}^{n^*}
\left[\frac{\Pr\left(\mA_k, \alpha_{k-1} < p_j \le \alpha_k\right)}{D_{-1}+k}  
-\frac{\Pr\left(\mA_k, p_j \le \alpha_{k-1}\right)}{(D_{-1}+k)(D_{-1}+k-1)} \right].
\end{equation}
Let $p_{(1)}^{(-j)} \le \cdots \le p_{( n^*-1)}^{(-j)}$ be the ordered $p$-values after excluding $p_j$. We denote the set $\left\{p^{(-j)}_{(1)}\le\alpha_1, \ldots, p^{(-j)}_{(k-1)}\le\alpha_{k-1}\right\}$ by $\mB_{k-1}^{(-j)}$, which represents the case that the first $(k-1)$ ordered $p$-values after excluding $p_j$ are each below the first $(k-1)$ thresholds. We note two facts that relate $\mA_k$ and $\mB_{k-1}^{(-j)}$. First, $\mA_k$ and $\left\{\alpha_{k-1} < p_j\right\}$ together imply that $p_j$ cannot be among the first $(k-1)$ smallest $p$-values and thus the first $(k-1)$ ordered $p$-values before and after excluding $p_j$ remain the same set. In addition, for any $j\in\{(k), \ldots, (n^*)\}$, $\left\{p_j \le \alpha_k\right\}$ implies $\left\{p_{(k)}\le \alpha_k\right\}$. Thus we have
$\Pr\left(\mA_k,  \alpha_{k-1} < p_j \le \alpha_k\right)= \Pr\left(\mB_{k-1}^{(-j)}, \alpha_{k-1} < p_j \le \alpha_k\right).
$
 Second, $\mB_{k-1}^{(-j)}$ is a subset of $\mA_k$ when $p_j\le\alpha_{k-1}$, which yields
$
 \Pr\left(\mA_k, p_j \le \alpha_{k-1}\right) \geq \Pr\left(\mB_{k-1}^{(-j)}, p_j \le \alpha_{k-1}\right).
$
By the two facts, we see that expression (\ref{p1:A}) is less than
\begin{equation*}
\sum_{j\in\mH_0}\sum_{k=1}^{n^*}
\left[\frac{\Pr\left(\mB_{k-1}^{(-j)}, \alpha_{k-1} < p_j \le \alpha_k\right)}{D_{-1}+k} 
- \frac{\Pr\left(\mB_{k-1}^{(-j)}, p_j \le \alpha_{k-1}\right)}{(D_{-1}+k)(D_{-1}+k-1)}\right],
\end{equation*}
which, by the independence assumption between $p_j$ and all other $p$-values given $\mathcal{G}_{l-1}$, becomes 
\begin{equation*}
\sum_{j \in \mH_0}\sum_{k=1}^{n^*}
\Pr\left(\mB_{k-1}^{(-j)}\right)
\times\left[\frac{\Pr\left(p_j  \le \alpha_k\right)}{D_{-1}+k}
-\frac{\Pr\left(p_j \le \alpha_{k-1}\right)}{D_{-1}+k-1} \right].
\end{equation*}
Applying Lemma 1 with $a_k=\Pr\left(p_j \le \alpha_k\right)/(D_{-1}+k)$ and $b_k=\Pr\left(\mB_{k-1}^{(-j)}, p^{(-j)}_{(k)} > \alpha_k\right)$, the foregoing expression reduces to
\begin{equation}
\label{p1:back_to_rewritten}
\sum_{j \in \mH_0}\sum_{k=1}^{n^*}
\Pr\left(\mB_{k-1}^{(-j)}, p^{(-j)}_{(k)} > \alpha_k\right)
\times \frac{\Pr\left(p_j  \le \alpha_k\right)}{D_{-1}+k}.
\end{equation}
Now we find \eqref{p1:back_to_rewritten} to be an upper bound for \eqref{p1:rewrite}.

By the uniform distribution of $p_j$ for $j\in\mathcal{H}_0$ and the definition of the thresholds $\{\alpha_k\}$, we have $\Pr\left(j\le\alpha_k\right)/(D_{-1}+k)=\alpha_k/(D_{-1}+k) \le q_l (1 - \alpha_k)/(n^*+1-k) = q_l \Pr\left(p_j>\alpha_k\right)/(n^*+1-k)$. In addition, replacing $\sum_{j \in \mH_0}$ by $\sum_{j =1}^{n^*}$ in (\ref{p1:back_to_rewritten}) yields an upper bound for (\ref{p1:back_to_rewritten}):
\begin{equation*}
q_l \sum_{k=1}^{n^*} 
(n^*+1-k)^{-1}
\sum_{j =1}^{n^*}\Pr\left(\mB_{k-1}^{(-j)}, p^{(-j)}_{(k)} > \alpha_k, p_j>\alpha_k\right).
\end{equation*}
We see that $\mB_{k-1}^{(-j)}$ and $\left\{p_j>\alpha_k\right\}$ imply that $p_j$ is not among the top $(k-1)$ smallest $p$-values. For either $j=(k)$ or $j\in \{(k+1), \ldots, (n^*)\}$, we infer from $p^{(-j)}_{(k)} > \alpha_k$ and $p_j>\alpha_k$ that $p_{(k)} > \alpha_k$. Therefore, $\Pr\left(\mB_{k-1}^{(-j)}, p^{(-j)}_{(k)} > \alpha_k, p_j>\alpha_k\right)=\Pr\left(\mA_{k-1}, p_{(k)} > \alpha_k\right)$ for $j\in\{(k), \ldots, (n^*)\}$ and $0$ for $j\in\{(1), \ldots, (k-1)\}$, which simplifies the above expression to
\begin{equation*}
q_l\sum_{k=1}^{n^*}\Pr\left(\mA_{k-1}, p_{(k)} > \alpha_k\right)
=q_l\left[1-\Pr\left(\mA_{n^*}\right)\right] 
\le q_l.
\end{equation*}
We complete the proof of \eqref{p1:key_step}.
\end{proof}

\bigskip

\noindent{\large\bf A.2 Proof of Theorem 2}

We first find the mathematical form for weight $\omega_j$ (after omiting the level index $l$ from the original notation $\omega_{l,j}$) for node $j$ at level $l$. Recall that the $p$-value at node $j$ is denoted by $p_j$. Let $T_1, T_2, \ldots, T_{t}$ represent all subtrees that contain node $j$. For example, the node $N_{1,1}$ in Figure \ref{tree}(a) is contained in subtrees rooted at $N_{2,1}$, $N_{3,1}$ and $N_{4,1}$, which are denoted by $T_1$, $T_2$ and $T_3$, respectively. Let $p_{T_s}^{(-j)}$ be the set of $p$-values of all level-$l$ nodes that are contained in subtree $T_s$ excluding node $j$. First, each weight $\omega_j$ always starts with $1$ indicating the node itself. Then for each subtree $T_s$, if $p_j$ is the (only) maximum $p$-value, i.e., $p_j>p_{T_s}^{(-j)}$, rejecting node $j$ will entail the root node of that subtree to also be rejected and thus add $1$ indicating that root node to $\omega_j$. Therefore, we can write $\omega_j$ as summation of a series of indicator functions:
\begin{equation}
\label{eq:weight_def}
\omega_j = 1 + \sum_{s=1}^{t}\bI\left(p_j>p_{T_s}^{(-j)}\right).
\end{equation}
Note that $\omega_j$ is dependent on $p$-values at level $l$ and thus considered to be random (even given the detection events below level $l$). 

Lemma 2 states a property of $\omega_j$, which will be useful in the proof of Theorem 2.
\begin{lemma}{2}
Suppose node $j$ has weight $\omega_j$ as defined in \eqref{eq:weight_def}. Assume that node $j$ is under the null hypothesis and so its $p$-value $p_j$ follows the uniform distribution. Also assume that $p_j$ is independent of all other $p$-values at level $l$. Let $\mB^{(-j)}$ denote the case that all other $p$-values excluding $p_j$ belong to a Borel set. For any $\alpha \in (0,1)$, we have
\begin{equation*}
\bE\left[\omega_j\bI\left(\mB^{(-j)}, p_j \le \alpha\right)\right]
\le \frac{\alpha}{1-\alpha}\bE\left[\omega_j\bI\left(\mB^{(-j)}, p_j > \alpha\right)\right].
\end{equation*}
\end{lemma}

\begin{proof}[\textit{Proof of Lemma 2}]
Assume that there is no tie in $p$-values. By the independence assumption of $p_j$ and the other $p$-values and the uniform distribution of $p_j$, we have 
$\bE\left[\bI\left(\mB^{(-j)}, p_j \le \alpha\right)\right]/\alpha
= \Pr\left(\mB^{(-j)}\right) = \bE\left[\bI\left(\mB^{(-j)}, p_j > \alpha\right)\right]/(1-\alpha)$. Due to the linearity of $\omega_j$, it then suffices to show for any subtree $T_s$ that
\begin{equation}
\label{l2:sufficient}
\frac{\bE\left[\bI\left(\mB^{(-j)}, p_j>p_{T_s}^{(-j)}, p_j \le \alpha\right)\right]}{\Pr\left(p_j \le \alpha\right)} \le 
\frac{\bE\left[\bI\left(\mB^{(-j)}, p_j>p_{T_s}^{(-j)}, p_j > \alpha\right)\right]}{\Pr\left(p_j > \alpha\right)}.
\end{equation}
By the mean value theorem, there exist $p_j^{*}$ and $p_j^{**}$, where $0<p_j^{*} < \alpha < p_j^{**}<1$, such that the left hand side of \eqref{l2:sufficient} becomes 
$$
\int_{0}^{\alpha}\Pr\left(\mB^{(-j)}, p_j> p_{T_s}^{(-j)}~\big|~p_j\right)dF(p_j)\Big/\int_{0}^{\alpha}1dF(p_j) 
= \Pr\left(\mB^{(-j)}, p_j^{*}>p_{T_s}^{(-j)}\right)
$$
and the right hand side becomes 
$$
\int_{\alpha}^{1}\Pr\left(\mB^{(-j)}, p_j>p_{T_s}^{(-j)}~\big|~p_j\right)dF(p_j)\Big/\int_{\alpha}^{1}1dF(p_j)
= \Pr\left(\mB^{(-j)}, p_j^{**}>p_{T_s}^{(-j)}\right).
$$
The fact that $p_j^{*} \le p_j^{**}$ gives us \eqref{l2:sufficient}.
\end{proof}

\vspace*{0.0cm}

\begin{proof}[\textit{Proof of Theorem 2}]

We introduce some new notation related to the weights $\omega_1\ldots,\omega_{n^*}$ for level-$l$ nodes. We denote the relative ordering of $p$-values $p_1,\ldots,p_{n^*}$ by $\mathcal{O}$. The weights defined in \eqref{eq:weight_def} are thus uniquely determined given the detection events at lower levels $\mathcal{G}_{l-1}$ as well as the ordering $\mathcal{O}$. Let $\omega_{[j]}$ be the weight corresponding to the $j$-th smallest $p$-value $p_{(j)}$ and $C_k = \sum_{j=1}^k \omega_{[j]}$ for $k=1,\ldots,n^*$. Thus, $C_k$ is also deterministic given $\{\mathcal{G}_{l-1}, \mathcal{O}\}$. Let $\omega_{(1)} \le \omega_{(2)} \le \cdots \le \omega_{(n^*)}$ denote the sorted weights by their own values; note that $\omega_{(j)}$ is often different from $\omega_{[j]}$. As illustrated in Section 2.4, $\omega_{(j)}$ is deterministic given $\mathcal{G}_{l-1}$ only under Condition (C1), regardless of the ordering of $p$-values. Denote $c_k=\sum_{j=1}^k\omega_{(j)}$ and $\overline{c}_k=\sum_{j=k}^{n^*}\omega_{(j)}$. Hence, $c_k \le C_k$, $\overline{c}_k \geq C_{n^*-k+1}$, and the threshold $\alpha_k$ satisfies
\begin{equation}
\label{p2:thresholds}
\frac{\alpha_k}{1-\alpha_k} \le \frac{D_{-1} + \sum_{j=1}^k \omega_{(j)} }{\sum_{j=k}^{n^*} \omega_{(j)}}q_l = \frac{D_{-1}+c_k}{\overline{c}_k}q_l.
\end{equation}

Like the proof of Theorem 1, the key step is to show that for every level $l$
\begin{equation}
\label{p2:key_step}
\bE \left[ \frac{V_l}{D_{l-1}+(R_l \bigvee 1)} \bigg| \mathcal{G}_{l-1}  \right] \le q_l.
\end{equation}
Again, we omit the index $l$ and the condition $\mathcal{G}_{l-1}$ and denote the set $\left\{p_{(1)}\le \alpha_1,\ldots,p_{(k)}\le \alpha_k\right\}$ by $\mA_k$. The left hand side of \eqref{p2:key_step} can be rewritten as
\begin{equation*}
\sum_{j\in\mH_0}\sum_{k=1}^{n^*}\bE\left[\frac{\omega_{j}\bI\left(\mA_k, p_{(k+1)} > \alpha_{k+1}, p_{j} \le \alpha_k\right) }{D_{-1}+C_k} \right], 
\end{equation*}
where $C_k$ by definition counts the number of all rejections that are entailed by rejecting the $k$-th smallest $p$-value. Replacing the expectation by double expectations that first conditions on $\mathcal{O}$ yields
\begin{equation}
\label{p2:rewrite}
\sum_{j\in\mH_0}\sum_{k=1}^{n^*}\bE_{\mathcal{O}}\left\{\omega_{j}\bE\left[\frac{\bI\left(\mA_k, p_{(k+1)} > \alpha_{k+1}, p_j \le \alpha_k\right) }{D_{-1}+C_k} \bigg| \mathcal{O} \right] \right\},
\end{equation}
where $C_k$ becomes a constant given $\mathcal{O}$. 

Once we condition on the order $\mathcal{O}$ and the weights $\{w_j\}$ become constants, similar arguments used in steps between expressions \eqref{p1:rewrite} and \eqref{p1:back_to_rewritten} in proof of Theorem 1 can also be used to obtain an upper bound for \eqref{p2:rewrite}:
\begin{equation}
\label{p2:back_to_rewritten}
    \sum_{j\in\mH_0}\sum_{k=1}^{n^*}\bE_{\mathcal{O}}\left\{\omega_{j} \bE\left[\frac{\bI\left(\mB_{k-1}^{(-j)}, p^{(-j)}_{(k)} > \alpha_k, p_{j} \le \alpha_k\right) }{D_{-1}+C_k} \Bigg| \mathcal{O} \right] \right\}.
    \end{equation}
Next, we combine the double expectations in \eqref{p2:back_to_rewritten} into one and use $c_k \le C_k$ to find that \eqref{p2:back_to_rewritten} is less than
\begin{equation}
\label{p2:back_to_rewritten_1}
 \sum_{j\in\mH_0}\sum_{k=1}^{n^*}\left(D_{-1}+c_k\right)^{-1}\bE\left[\omega_{j}\bI\left(\mB_{k-1}^{(-j)}, p^{(-j)}_{(k)} > \alpha_k, p_{j} \le \alpha_k\right) \right].
    \end{equation}
    
Now the weight $\omega_j$ is considered random again. According to Lemma 2 with $\mB^{(-j)}=\left\{\mB_{k-1}^{(-j)}, p^{(-j)}_{(k)} > \alpha_k\right\}$, $\alpha=\alpha_k$, and a null $p$-value $p_j$, we obtain 
\begin{equation*}
\bE\left[ \omega_{j}\bI\left(\mB_{k-1}^{(-j)}, p^{(-j)}_{(k)} > \alpha_k, p_{j} \le \alpha_k\right)   \right] \le \frac{\alpha_k}{1-\alpha_k} \bE\left[ \omega_{j}\bI\left(\mB_{k-1}^{(-j)}, p^{(-j)}_{(k)} > \alpha_k, p_{j} > \alpha_k\right)   \right].
\end{equation*}
Using \eqref{p2:thresholds} and replacing $\sum_{j\in\mH_0}$ by $\sum_{j=1}^{n^*}$, we see (\ref{p2:back_to_rewritten_1}) is less than 
\begin{equation*}
q_l\sum_{k=1}^{n^*} \left(\overline{c}_k\right)^{-1}\bE\left[\sum_{j=1}^{n^*}\omega_{j}\bI\left(\mB_{k-1}^{(-j)}, p^{(-j)}_{(k)} > \alpha_k, p_{j} > \alpha_k\right) \right].
\end{equation*}
By the same arguments as in proof of Theorem 1, $\bI\left(\mB_{k-1}^{(-j)}, p^{(-j)}_{(k)} > \alpha_k, p_j>\alpha_k\right)=\bI\left(\mA_{k-1}, p_{(k)} > \alpha_k\right)$ for $j\in\{(k), \ldots, (n^*)\}$ and $0$ for $j\in\{(1), \ldots, (k-1)\}$. Therefore, the above expression simplifies to 
$$
q_l\sum_{k=1}^{n^*} \left(\overline{c}_k\right)^{-1}\bE\left[\overline{s}_k \bI\left(\mA_{k-1}, p_{(k)} > \alpha_k\right)\right]
\le q_l \sum_{k=1}^{n^*}  \Pr\left(\mA_{k-1}, p_{(k)} > \alpha_k\right) \nonumber
= q_l\left[1-\Pr\left(\mA_{n^*}\right)\right] 
\le q_l.
$$
We complete the proof of \eqref{p2:key_step}.
\end{proof}

\bigskip

\noindent{\large\bf A.3 Least favorable weights}

We obtain $\widetilde{\boldsymbol\omega}=(\widetilde{\omega}_{l,(1)},\cdots,\widetilde{\omega}_{l,(n_l^*)})$ in a recursive manner. Recall that a weight $\omega_{l,j}$ counts the number of nodes that are simultaneously rejected if node $j$ is rejected, which includes node $j$ and some of its ancestors. Let $\widetilde{\boldsymbol{\omega}}^{(l+h)}$ $(h=0,1,2,\ldots)$ denote the least favorable set of weights against any possible set of weights $\boldsymbol{\omega}^{(l+h)}$ when nodes between level $l$ and level $(l+h)$ (including levels $l$ and $l+h$) are counted. Trivially, $\widetilde{\boldsymbol{\omega}}^{(l)}=\boldsymbol{\omega}^{(l)}=(1,\ldots,1)$ with 1 for every node when only nodes at level $l$ are counted, and this serves as the starting point of the recursive algorithm. At the root node level, denoted by $(l+h^*)$, $\widetilde{\boldsymbol{\omega}}^{(l+h^*)}$ is the least favorable set of weights $\widetilde{\boldsymbol{\omega}}$ that we wish to obtain, and is the end point of the algorithm. We find that $\widetilde{\boldsymbol{\omega}}^{(l+h)}$ can be derived from $\widetilde{\boldsymbol{\omega}}^{(l+h-1)}$ by traversing over every node at level $l+h$, locating the subset of elements in $\widetilde{\boldsymbol{\omega}}^{(l+h-1)}$ that correspond to the level-$l$ descendants of that node, and adding 1 to the largest element of that subset; if there exist multiple largest elements, randomly pick one and add 1. For example, to calculate least favorable weights for the bottom level of the tree in Figure \ref{tree}(b), we obtain $\widetilde{\boldsymbol{\omega}}^{(l)}=(1,1,1,1,1,1)$, $\widetilde{\boldsymbol{\omega}}^{(l+1)}=(1,2,1,1,1,2)$, and $\widetilde{\boldsymbol{\omega}}^{(l+2)}=(1,2,1,1,1,3)$, and ultimately the sorted version $\widetilde{\boldsymbol{\omega}}=(1,1,1,1,2,3)$. Note that we ordered the individual weights in intermediate $\widetilde{\boldsymbol{\omega}}^{(l+h)}$ by the physical position of the bottom-level nodes in the displayed tree, and only sort the weights in the last step. This procedure guarantees that every $\widetilde{\boldsymbol{\omega}}^{(l+h)}$ ($h=0,1,\ldots,h^*$) is the least favorable against any arbitrary $\boldsymbol{\omega}^{(l+h)}$ in which the count 1 is added to one element other than the largest one for at least one subset. It also ensures the uniqueness of sorted $\widetilde{\boldsymbol{\omega}}$, which leads to a unique set of thresholds. Finally, the least favorable weights $\widetilde{\boldsymbol{\omega}}$ corresponds to the ordering of $p$-values that is equal to the ordering of depths at level-$l$ nodes, e.g., the node with the largest $p$-value is the node with the largest depth.

\bigskip

\noindent{\large\bf A.4 Proof of Theorem 3}

\begin{proof}[\textit{Proof of Theorem 3}]
Let $\widetilde{\omega}_{(1)} \le\widetilde{\omega}_{(2)} \le \cdots \le \widetilde{\omega}_{(n^*)}$ denote sorted least favorable weights after omitting the level index $l$. Denote $\widetilde{c}_k = \sum_{j=1}^k \widetilde{\omega}_{(j)}$ and  $\overline{\widetilde{c}}_k = \sum_{j=k}^{n^*} \widetilde{\omega}_{(j)}$. The same arguments in the proof of Theorem 2 can be used except that we use $\widetilde{c}_k$ in place of $c_k$, $\overline{\widetilde{c}}_k$ in place of $\overline{c}_k$, and the thresholds \eqref{threshold3} in place of thresholds \eqref{threshold2}.
\end{proof}

\bigskip

\noindent{\large\bf A.5 Proof of Theorem 4}

\begin{proof}[\textit{Proof of Theorem 4}]
The OTU-level testing in the two-stage procedure is exactly the same as the OTU-level testing in the one-stage procedure, so we immediately have $\mathrm{FAR}_{\text{otu}} \le q_1$. Then we rewrite the FAP among all taxa (inner) levels as
$$
\mathrm{FAP}_{\text{taxa}} = \frac{\sum_{l=2}^LV_l}{(\sum_{l=2}^LR_l)\bigvee 1}\le \sum_{l=2}^L\frac{V_l}{(\sum_{l'=2}^lR_{l'})\bigvee 1}=\sum_{l=2}^L\frac{V_l}{(D_{l-1}^{\dagger}+R_l)\bigvee 1}=\sum_{l=2}^L\frac{V_l}{D_{l-1}^{\dagger}+(R_l\bigvee 1)}.
$$
We note that $D_{l-1}^{\dagger}$ is deterministic conditioning on the detection events at lower levels, denoted by $\mathcal{G}_{l-1}^{\dagger}$, which excludes the OTU level. Then we follow the same steps in the proofs of Theorems 2 and 3, replacing $D_{l-1}$ by $D_{l-1}^{\dagger}$ and $\mathcal{G}_{l-1}$ by $\mathcal{G}_{l-1}^{\dagger}$ to obtain 
$\bE \left[ V_l/\left\{D_{l-1}^{\dagger}+(R_l \bigvee 1)\right\} \Big| \mathcal{G}_{l-1}^{\dagger}  \right] \le q_l$ for $l = 2,  \ldots, L$. Finally, 
$$
\text{FAR}_{\text{taxa}}=\bE\left(\mathrm{FAP}_{\text{taxa}} \right)\le\sum_{l=2}^L \bE\left\{ \bE\left[\frac{V_{l}}{D_{l-1}^{\dagger}+(R_l\bigvee 1)}\bigg|\mathcal{G}_{l-1}^{\dagger} \right]\right\}\le \sum_{l=2}^Lq_l = q_{-1},
$$
which implies that FAR among all taxa levels are controlled by $q_{-1}$. Indeed, this is the same as applying the one-stage testing at FAR $q_{-1}$ to the subtree after removing the whole OTU level and the higher-level taxa that are detected because all of their corresponding OTUs are detected.
\end{proof}

\bibliographystyle{jasa.bst}
\bibliography{literature}

\newpage

\section*{Supplementary Materials}

\renewcommand{\thefigure}{S\arabic{figure}}

\setcounter{figure}{0}

\renewcommand{\thetable}{S\arabic{table}}

\setcounter{table}{0}

\setcounter{page}{1}

\begin{figure}[H]
\centering
\includegraphics[width=0.7\textwidth, angle=90]{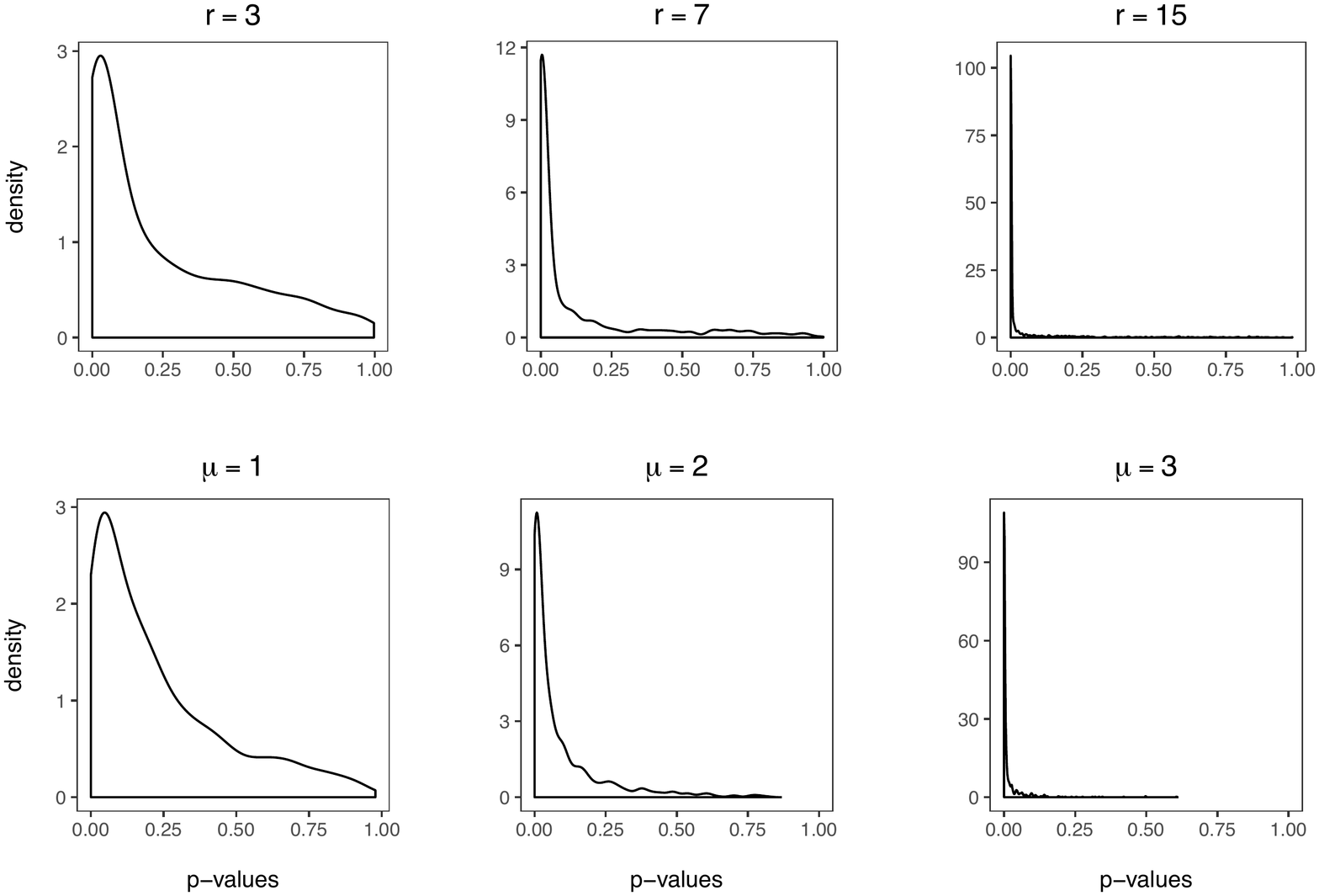}
\caption{\label{simulated_pvalue} Density functions of $p$-values generated from the Beta distribution $\mathrm{Beta}(1/r, 1)$ (upper panel) and the Gaussian-tailed model (lower panel). The two plots in each column have comparable ``height" at around zero but the upper ones always have heavier right tails.}
\end{figure}

\begin{figure}[H]
\centering
\includegraphics[width=0.9\textwidth]{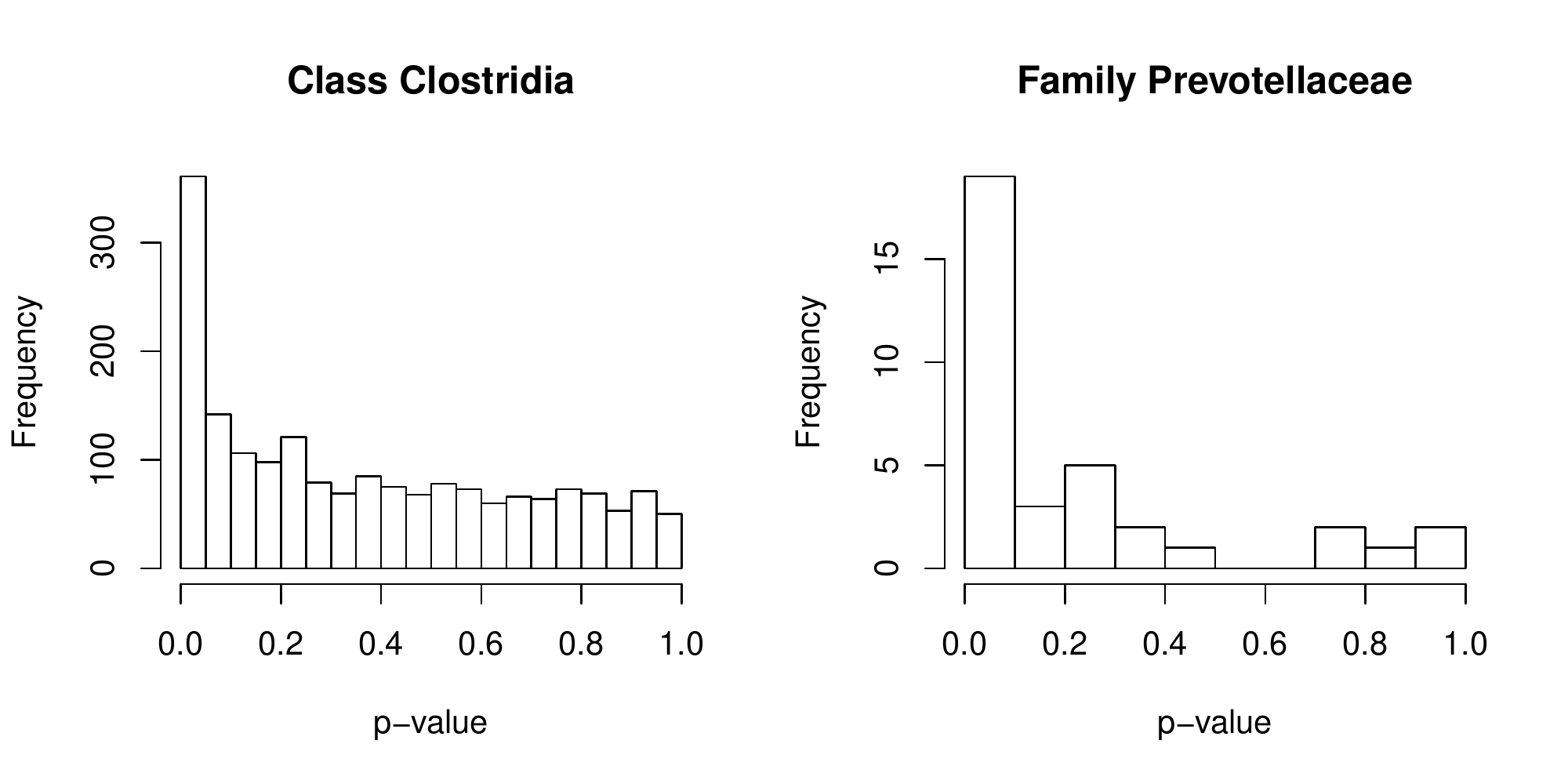}
\caption{\label{empirical_pvalue} Empirical distributions of $p$-values for OTUs in class \textit{Clostridia} and family \textit{Prevotellaceae} in the IBD data. These two taxa were detected to be driver taxa by the weighted bottom-up test and contain the most OTUs.}
\end{figure}

\begin{figure}[H]
\centering
\includegraphics[width=1.0\textwidth]{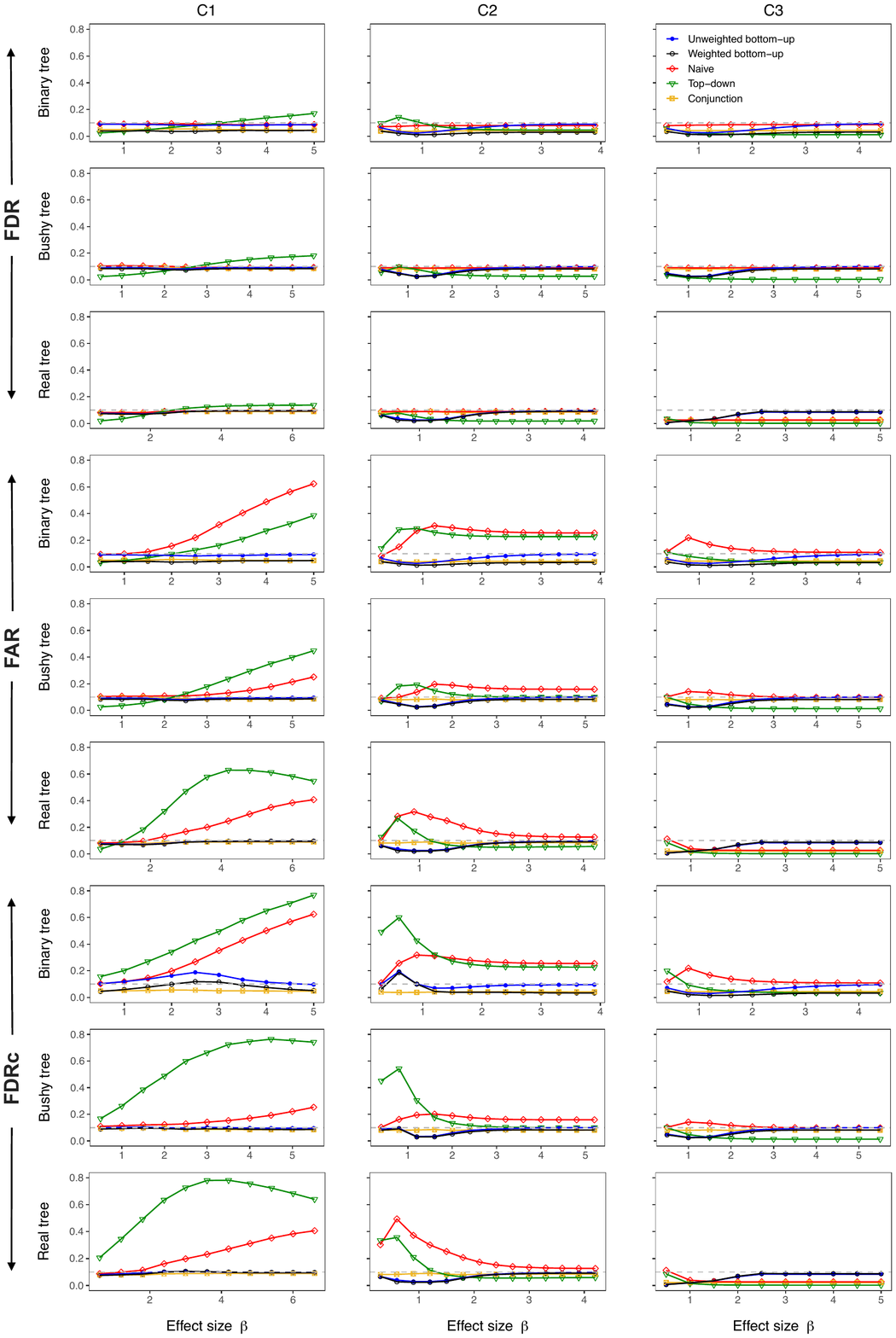}
\caption{\label{error_rates_gaussian} Error rates for testing all nodes in the tree. The non-null $p$-values at leaf nodes were simulated from the Gaussian-tailed model. }
\end{figure}

\begin{figure}[H]
\centering
\includegraphics[width=0.8\textwidth]{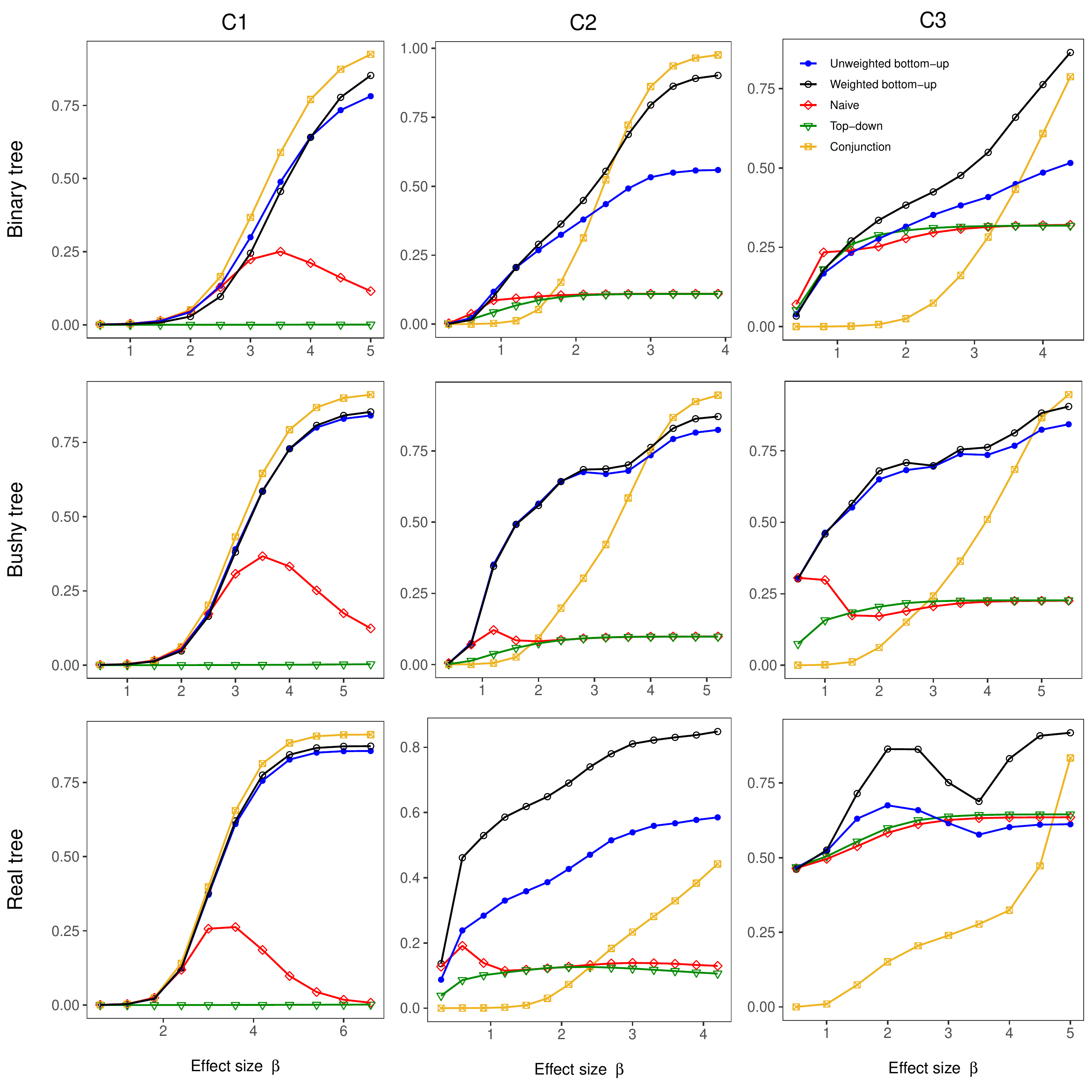}
\caption{\label{sensitivity_gaussian} Accuracy (weighted Jaccard similarity) for detecting all associated nodes (including the driver nodes and all of their descendants at all levels). The non-null $p$-values at leaf nodes were simulated from the Gaussian-tailed model. }
\end{figure}

\begin{figure}[H]
\centering
\includegraphics[width=0.8\textwidth]{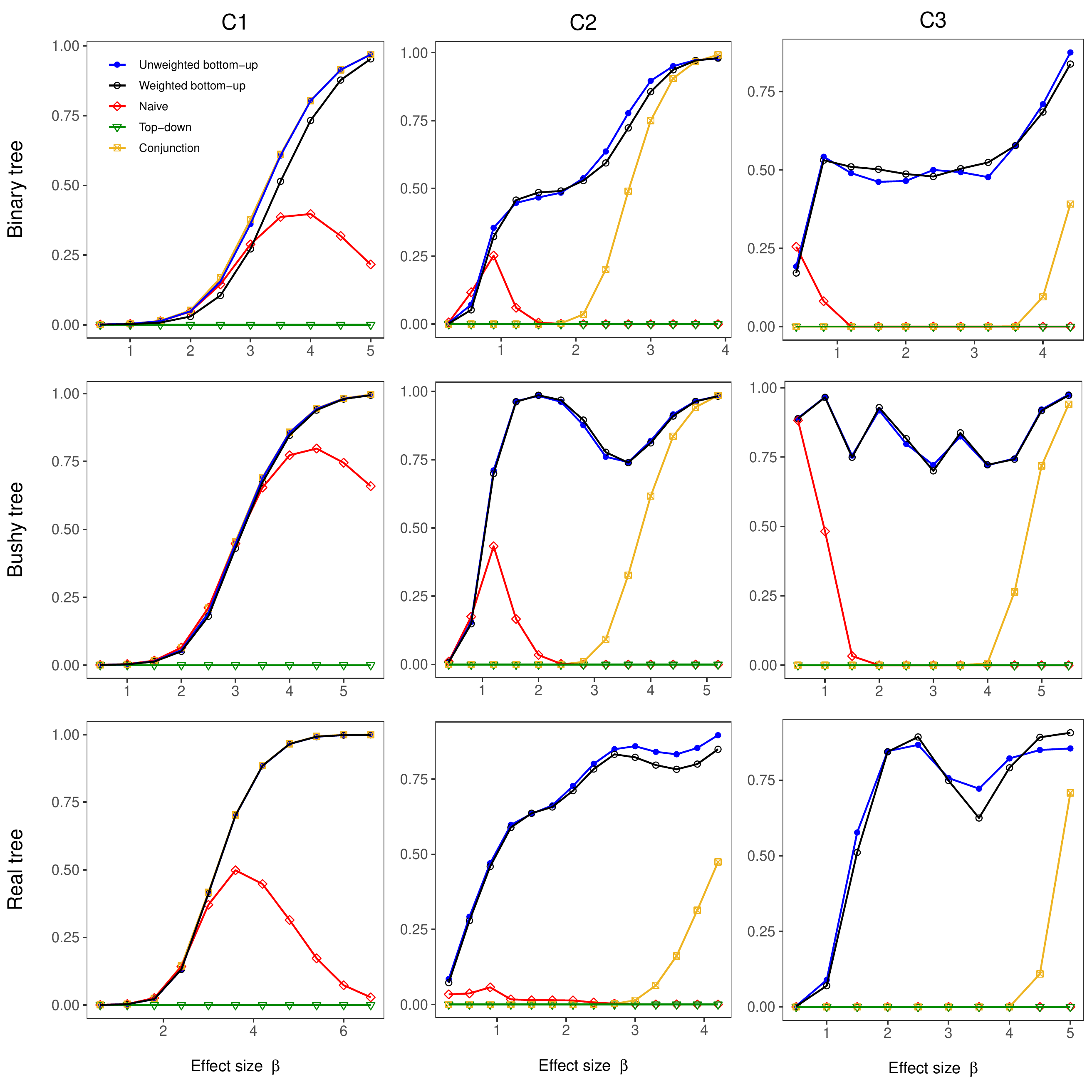}
\caption{\label{driver_nodes_gaussian} Percentage of driver nodes that were pinpointed. The non-null $p$-values at leaf nodes were simulated from the Gaussian-tailed model.}
\end{figure}

\begin{sidewaystable}[htbp]
  \centering
  \caption{Taxa detected by the weighted bottom-up test to be differential abundant between the UC and control groups}
  \scriptsize
    \begin{tabular}{ll}   
    \toprule
       Taxon & \multicolumn{1}{p{4.19em}}{$p$-value} \\
    \midrule
\textit{p\_Actinobacteria}; \textit{c\_Coriobacteriia}\dr& $4.04\times10^{-4}$ \\
\textit{p\_Actinobacteria}; \textit{c\_Coriobacteriia}\dr; \textit{o\_Coriobacteriales} & $4.04\times10^{-4}$ \\
\textit{p\_Actinobacteria}; \textit{c\_Coriobacteriia}\dr; \textit{o\_Coriobacteriales}; \textit{f\_Coriobacteriaceae}& $4.04\times10^{-4}$ \\
\textit{p\_Actinobacteria}; \textit{c\_Coriobacteriia}\dr; \textit{o\_Coriobacteriales}; \textit{f\_Coriobacteriaceae}; \textit{g\_Slackia}& $4.61\times10^{-4}$ \\
\textit{p\_Bacteroidetes}; \textit{c\_Bacteroidia}; \textit{o\_Bacteroidales}; \textit{f\_Bacteroidaceae}; \textit{g\_Bacteroides}; \textit{s\_ovatus}\dr & $2.57\times10^{-7}$ \\
\textit{p\_Bacteroidetes}; \textit{c\_Bacteroidia}; \textit{o\_Bacteroidales}; \textit{f\_[Paraprevotellaceae]}; \textit{g\_[Prevotella]}\dr& $8.57\times10^{-4}$ \\
\textit{p\_Bacteroidetes}; \textit{c\_Bacteroidia}; \textit{o\_Bacteroidales}; \textit{f\_Prevotellaceae}\dr& $1.95\times10^{-3}$ \\
\textit{p\_Bacteroidetes}; \textit{c\_Bacteroidia}; \textit{o\_Bacteroidales}; \textit{f\_Prevotellaceae}\dr; \textit{g\_Prevotella}& $1.95\times10^{-3}$ \\
\textit{p\_Bacteroidetes}; \textit{c\_Bacteroidia}; \textit{o\_Bacteroidales}; \textit{f\_Prevotellaceae}\dr; \textit{g\_Prevotella}; \textit{s\_copri}   & $1.44\times10^{-9}$ \\
\textit{p\_Bacteroidetes}; \textit{c\_Bacteroidia}; \textit{o\_Bacteroidales}; \textit{f\_S24-7}\dr& $2.10\times10^{-3}$ \\
\textit{p\_Cyanobacteria}; \textit{c\_Chloroplast}\dr& $1.47\times10^{-4}$ \\
\textit{p\_Cyanobacteria}; \textit{c\_Chloroplast}\dr; \textit{o\_Streptophyta}& $1.47\times10^{-4}$ \\
\textit{p\_Firmicutes}; \textit{c\_Clostridia}\dr& $7.77\times10^{-4}$ \\
\textit{p\_Firmicutes}; \textit{c\_Clostridia}\dr; \textit{o\_Clostridiales} & $< 1\times10^{-16}$ \\
\textit{p\_Firmicutes}; \textit{c\_Clostridia}\dr; \textit{o\_Clostridiales}; \textit{f\_Clostridiaceae}; \textit{g\_Clostridium}& $5.80\times10^{-5}$ \\
\textit{p\_Firmicutes}; \textit{c\_Clostridia}\dr; \textit{o\_Clostridiales}; \textit{f\_Lachnospiraceae}& $2.23\times10^{-13}$ \\
\textit{p\_Firmicutes}; \textit{c\_Clostridia}\dr; \textit{o\_Clostridiales}; \textit{f\_Lachnospiraceae}; \textit{g\_Coprococcus}& $5.27\times10^{-7}$ \\
\textit{p\_Firmicutes}; \textit{c\_Clostridia}\dr; \textit{o\_Clostridiales}; \textit{f\_Ruminococcaceae}& $< 1\times10^{-16}$ \\
\textit{p\_Firmicutes}; \textit{c\_Clostridia}\dr; \textit{o\_Clostridiales}; \textit{f\_Ruminococcaceae}; \textit{g\_Faecalibacterium}; \textit{s\_prausnitzii} & $2.02\times10^{-4}$ \\
\textit{p\_Firmicutes}; \textit{c\_Clostridia}\dr; \textit{o\_Clostridiales}; \textit{f\_Veillonellaceae}; \textit{g\_Veillonella}; \textit{s\_parvula} & $1.99\times10^{-4}$ \\
\textit{p\_Firmicutes}; \textit{c\_Erysipelotrichi}\dr& $1.37\times10^{-5}$ \\
\textit{p\_Firmicutes}; \textit{c\_Erysipelotrichi}\dr; \textit{o\_Erysipelotrichales}& $1.37\times10^{-5}$ \\
\textit{p\_Firmicutes}; \textit{c\_Erysipelotrichi}\dr; \textit{o\_Erysipelotrichales}; \textit{f\_Erysipelotrichaceae} & $1.37\times10^{-5}$ \\
\textit{p\_Proteobacteria}; \textit{c\_Gammaproteobacteria}; \textit{o\_Enterobacteriales}; \textit{f\_Enterobacteriaceae}; \textit{g\_Morganella}\dr& $1.39\times10^{-3}$ \\
\textit{p\_Proteobacteria}; \textit{c\_Gammaproteobacteria}; \textit{o\_Enterobacteriales}; \textit{f\_Enterobacteriaceae}; \textit{g\_Enterobacter}; \textit{s\_radicincitans}\dr & $5.63\times10^{-4}$ \\
\textit{p\_Tenericutes}; \textit{c\_RF3}\dr& $2.81\times10^{-3}$ \\
\textit{p\_Verrucomicrobia}\dr& $8.48\times10^{-5}$ \\
\textit{p\_Verrucomicrobia}\dr; \textit{c\_Verrucomicrobiae}&$8.48\times10^{-5}$ \\
\textit{p\_Verrucomicrobia}\dr; \textit{c\_Verrucomicrobiae}; \textit{o\_Verrucomicrobiales} & $8.48\times10^{-5}$ \\
\textit{p\_Verrucomicrobia}\dr; \textit{c\_Verrucomicrobiae}; \textit{o\_Verrucomicrobiales}; \textit{f\_Verrucomicrobiaceae}& $8.48\times10^{-5}$ \\
\textit{p\_Verrucomicrobia}\dr; \textit{c\_Verrucomicrobiae}; \textit{o\_Verrucomicrobiales}; \textit{f\_Verrucomicrobiaceae}; \textit{g\_Akkermansia}& $8.48\times10^{-5}$ \\

    \bottomrule
    \end{tabular}%
  \label{tab:uchc}%
      \\\flushleft NOTE: The detected driver taxa are marked with ``\#". Kingdom \textit{Bacteria} is omitted from the taxon names. 
\end{sidewaystable}%

\end{document}